\documentclass[journal]{IEEEtran}
\usepackage{balance}
\usepackage{graphicx}
\usepackage{bm}
\usepackage{amsmath}
\usepackage{amsfonts}
\usepackage{amsthm}
\usepackage{amssymb}
\usepackage{color}
\usepackage{mathrsfs}
\usepackage{upgreek}
\usepackage{flushend}
\usepackage[utf8]{inputenc}
\usepackage{tikz}
\usetikzlibrary{arrows, calc, decorations.markings, positioning}
\usepackage{cite}
\usepackage{booktabs}
\usepackage{multirow}
\usepackage{amsmath}
\usepackage{amsthm}
\usepackage{algorithm}
\usepackage{algorithmic}
\usepackage{epstopdf}
\usepackage[flushleft]{threeparttable}
\usepackage{midfloat}
\usepackage{tabularx}
\usepackage{threeparttable}
\ifCLASSOPTIONcompsoc
 \usepackage[caption=false,font=normalsize,labelfont=sf,textfont=sf]{subfig}
\else
 \usepackage[caption=false,font=footnotesize]{subfig}
\fi
\newtheorem{thm}{Theorem}
\newtheorem{assumption}{Assumption}
\newtheorem{definition}{Definition}
\newtheorem{lemma}{Lemma}
\newtheorem{corollary}{Corollary}

\usepackage{hyperref}					
\hypersetup{colorlinks,
	linkcolor=blue,%
	anchorcolor=blue,
    citecolor=blue}

\begin{document}
\title{\huge{Uncovering the Iceberg in the Sea: Fundamentals of Pulse Shaping and Modulation Design for Random ISAC Signals}
}
\author{
	{
	Fan Liu,~\IEEEmembership{Senior Member,~IEEE}, Yifeng Xiong,~\IEEEmembership{Member,~IEEE}, Shihang Lu,~\IEEEmembership{Student Member,~IEEE}, \\Shuangyang Li,~\IEEEmembership{Member,~IEEE}, Weijie Yuan,~\IEEEmembership{Senior Member,~IEEE}, Christos Masouros,~\IEEEmembership{Fellow,~IEEE}, \\Shi Jin,~\IEEEmembership{Fellow,~IEEE}, and~Giuseppe Caire,~\IEEEmembership{Fellow,~IEEE}
	} 
\thanks{(\textit{Corresponding author: Yifeng Xiong})}
\thanks{F. Liu and S. Jin are with the National Mobile Communications Research Laboratory, Southeast University, Nanjing 210096, China (e-mail: f.liu@ieee.org, jinshi@seu.edu.cn).}
\thanks{S. Lu and W. Yuan is with the School of System Design and Intelligent Manufacturing, Southern University of Science and Technology, Shenzhen 518055, China. (e-mail: lush2021@mail.sustech.edu.cn, yuanwj@sustech.edu.cn).}
\thanks{Y. Xiong is with the School of Information and Electronic Engineering, Beijing University of Posts and Telecommunications, Beijing 100876, China. (e-mail: yifengxiong@bupt.edu.cn)}
\thanks{S. Li and G. Caire are with the Chair of Communications and Information Theory, Technical University of Berlin, 10623 Berlin, Germany (e-mail: shuangyang.li@tu-berlin.de, caire@tu-berlin.de).}
\thanks{C. Masouros is with the Department of Electrical and Electronic Engineering, University College London, London WC1E 7JE, UK (email: c.masouros@ucl.ac.uk).}

}
\maketitle

\begin{abstract}
Integrated Sensing and Communications (ISAC) is expected to play a pivotal role in future 6G networks. To maximize time-frequency resource utilization, 6G ISAC systems must exploit data payload signals, that are inherently random, for both communication and sensing tasks. This paper provides a comprehensive analysis of the sensing performance of such communication-centric ISAC signals, with a focus on modulation and pulse shaping design to reshape the statistical properties of their auto-correlation functions (ACFs), thereby improving the target ranging performance. We derive a closed-form expression for the expectation of the squared ACF of random ISAC signals, considering arbitrary modulation bases and constellation mappings within the Nyquist pulse shaping framework. The structure is metaphorically described as an ``iceberg hidden in the sea", where the ``iceberg'' represents the squared mean of the ACF of random ISAC signals, that is determined by the pulse shaping filter, and the ``sea level'' characterizes the corresponding variance, caused by the randomness of the data payload. Our analysis shows that, for QAM/PSK constellations with Nyquist pulse shaping, Orthogonal Frequency Division Multiplexing (OFDM) achieves the lowest ranging sidelobe level across all lags. Building on these insights, we propose a novel Nyquist pulse shaping design to enhance the sensing performance of random ISAC signals. Numerical results validate our theoretical findings, showing that the proposed pulse shaping significantly reduces ranging sidelobes compared to conventional root-raised cosine (RRC) pulse shaping, thereby improving the ranging performance.
\end{abstract}

\section{Introduction}
\IEEEPARstart{T}{he} 6G wireless networks represent a groundbreaking paradigm shift, designed to support and enable a wide array of advanced technologies, including autonomous driving, smart manufacturing, digital twins, and applications within the emerging low-altitude economy \cite{Chafii2023CST,saad2019vision}. In contrast to previous generations that primarily focused on enhancing communication capabilities, 6G is expected to broaden its functional scope by incorporating novel technologies such as ISAC \cite{9737357}. This convergence of sensing and communication within a unified framework is poised to fundamentally reshape how wireless networks operate, enhancing both spectral and hardware efficiencies while creating a multi-functional network offering diverse services to a vast number of users. Notably, the International Telecommunication Union (ITU) recently acknowledged ISAC as one of the six key usage scenarios in its global 6G vision \cite{ITU2023}.

ISAC aims to facilitate the shared utilization of wireless resources—such as time, frequency, and power—for both sensing and communication functions through a unified hardware system \cite{9585321}. A significant challenge in this domain is the design of dual-functional signals capable of simultaneously handling target detection and data transmission within the ISAC channel \cite{10334037,10292797}. Current design paradigms can be classified into three primary approaches: sensing-centric, communication-centric, and joint designs \cite{9737357}. The sensing-centric approach integrates communication capabilities into classic radar signals, such as chirp waveforms. In contrast, the communication-centric method modifies existing communication protocols and signals to incorporate sensing functionalities. The joint design philosophy, however, advocates for the creation of entirely new waveforms that flexibly balance and integrate both sensing and communication objectives \cite{10018908}.

Each of the aforementioned designs addresses specific application needs, while the communication-centric approach is currently in focus for commercial 6G deployment due to its alignment with the standards and back compatibility to legacy systems \cite{10012421,9921271}. By employing the data payload of the transmitted frame directly for sensing tasks, this methodology eliminates the need for complex waveform re-design or adaptation. This simplification not only reduces system complexity but also enhances resource efficiency, positioning it as a compelling option for implementing ISAC in 6G networks. In contrast to conventional radar waveforms, which are endowed with deterministic structures, information entropy is at the core of communication signaling, necessitating an inherent degree of \textit{randomness} to effectively convey information. This intrinsic randomness, however, may compromise the sensing performance, imposing a challenge characterized as the ``Deterministic-Random Tradeoff (DRT)'' \cite{10147248}. Therefore, it is imperative to explore the achievable sensing performance of standard communication signals, and, more essentially, to clarify how this performance is shaped by the basic building blocks of communication systems.

A practical communication signal, at its most fundamental level, can be decomposed into three primary components: 1) \textit{Constellation symbols} mapped from information bits; 2) \textit{An orthonormal modulation basis} that conveys these symbols; and 3) \textit{A pulse shaping filter} that converts the time-domain discrete samples into a continuous form. All of them may have a significant impact on the resulting sensing performance. For instance, it is well-recognized in the radar community that signals with constant modulus enhance the sensing capabilities. This has been confirmed by recent studies on OFDM-based ISAC systems, which demonstrated that PSK constellations yield considerably lower ranging sidelobe levels compared to their QAM counterparts \cite{zhang2023input,Keskin2024fundamental}. Additionally, research has shown that the variance of the ambiguity function (AF) for OFDM signals is critically dependent on the fourth moment of the constellation, known as the \textit{kurtosis} \cite{du2024reshaping}. Specifically, OFDM signals modulated with high-kurtosis constellations exhibit significant variability in sidelobe levels of the AF, resulting in missed detection of weak targets and false detection of ghost targets, thus impairing the overall sensing performance. Therefore, it is important to control the kurtosis of constellations while maximizing their communication rates, which may be addressed through dedicated probabilistic constellation shaping (PCS) techniques \cite{du2024reshaping}.

In addition to constellation design, the choice of the modulation basis—often referred to as a ``waveform'' in a broader context—also significantly influences the ISAC transmission performance. The foundational study by Sturm and Wiesbeck \cite{sturm2011waveform} examined the effectiveness of OFDM signals in accurately measuring the delay and Doppler characteristics of radar targets. Since then, extensive research has focused on evaluating the sensing performance of various communication modulation schemes, including OFDM \cite{sturm2011waveform,zhang2023input,9529026, Keskin2024fundamental}, single-carrier (SC) \cite{9005192}, code division multiple access (CDMA) \cite{9724170}, orthogonal time-frequency space (OTFS) \cite{9109735,10463758}, and their combinations \cite{9359665,10264814}. Recent discussions have particularly focused on comparing the sensing capabilities of OTFS and OFDM. The analysis in \cite{9109735} revealed that OFDM demonstrates a slight performance gain over OTFS in terms of range and velocity estimation accuracy. In contrast, the findings in \cite{10638525} suggested that OTFS produces lower sidelobes in the AF for both delay and Doppler domains; however, this comparison was somewhat biased due to the use of random QPSK symbols for OFDM, while OTFS employed deterministic symbols. More recently, a novel modulation technique known as affine frequency division multiplexing (AFDM) has been introduced for high-mobility communications \cite{10087310,10439996}. This method strategically places symbols within the affine Fourier transform (AFT) domains, using orthogonal chirp signals as the signaling basis. AFDM is ainticipated to enhance the communication perofrmance in high-mobility scenarios due to its flexibility in designing the basis functions across the time-frequency-delay-Doppler domains, as well as the sensing performance due to the advantageous radar properties of chirp signals. 

More relevant to this study, a closed-form expression for the expectation of the squared auto-correlation function (ACF) for random ISAC signals (also referred to as the zero-Doppler cut of the AF) was derived in \cite{liu2024ofdm} under arbitrary modulation bases and rotationally symmetric constellations. In \cite{liu2024ofdm}, it was rigorously demonstrated that for QAM and PSK constellations, OFDM stands out as the only globally optimal modulation scheme, achieving the lowest average ranging sidelobe level among all communication signals with cyclic prefix (CP), while also being locally optimal among signals without CP. Although these findings help clarify the ongoing discourse regarding the most effective communication modulation scheme for realizing the sensing functionality, the framework presented in \cite{liu2024ofdm} is built upon Nyquist sampling and overlooks the impact of pulse shaping filters, leading to an incomplete picture of the overall sensing performance. Nyquist sampling, although sufficient for communication, is inadequate for analyzing the sensing performance of continuous-time ISAC signals, particularly when assessing ranging performance at fractional delays relative to symbol duration \cite{10118873}. In this context, the pulse shaping filter is crucial for determining the ACF and AF of random ISAC signals. The recent work in \cite{liao2024pulse} investigated the sensing performance of pulse-shaped SC signals, and introduced a randomness-aware pulse shaping design aimed at minimizing the average sidelobe level of the AF. However, this methodology is confined to SC signaling, and encounters considerable challenges when applying to arbitrary modulation bases. Moreover, the optimal modulation basis for sensing in the context of pulse shaping remains unclear, highlighting a critical gap in evaluating the performance of continuous-time communication-centric ISAC signals.

This paper offers an in-depth and systematic examination of the ranging performance of pulse-shaped random ISAC signals. We build upon and substantially extend the analytical framework introduced in previous research \cite{liu2024ofdm,liao2024pulse}, to apply to a broader class of signals that employ arbitrary orthonormal modulation bases and Nyquist pulse shaping filters, while considering the transmission of independent and identically distributed (i.i.d.) symbol sequences realized from any valid constellation. Specifically, we employ the periodic ACF of random ISAC signals as a performance indicator for ranging, which corresponds to the matched-filtering (MF) operation, and analyze its statistical properties. For clarity, our contributions are summarized as follows:
\begin{itemize}
    \item We derive a closed-form expression for the expectation of the squared ACF of random ISAC signals, revealing its structure within the framework of Nyquist pulse shaping. This structure can be metaphorically represented as an ``iceberg in the sea'', where the ``iceberg'' signifies the squared mean, critical in shaping the MF performance in ranging, and the ``sea level'' represents the variance of the ACF, whose increasing values would deteriorate the multi-target estimation accuracy. We demonstrate that the ``iceberg'' corresponds to the squared ACF of the chosen pulse, while the ``sea level'', determined by the variability of the random data, can be diminished by a factor of $M$ via coherently integrating the MF output by $M$ times.
    \item We prove that under QAM/PSK constellations, OFDM is the only modulation basis that achieves the lowest ranging sidelobe level at every lag for all Nyquist pulse shaping filters. This finding is a significant extension compared to previous works \cite{liu2024ofdm}, which solely focused on discrete time-domain samples and overlooked the effects of pulse shaping, and allows us to provide a more comprehensive understanding of communication-centric ISAC transmission.
    \item Building on the above understanding, we propose a new pulse shaping design aimed at enhancing the sensing performance of random ISAC signals. In particular, the coherent integration operation reduces only the ``sea level'' part. After integrating a sufficiently large number of MF outputs, the sensing performance will depend mainly on the geometry of the ``iceberg''. Based on this insight, we propose an optimization approach for designing Nyquist pulses to minimize the sidelobe level of the ACF within a specified delay region, referred to as the ``iceberg shaping'' technique.
\end{itemize}

The remainder of this paper is organized as follows. Sec. \ref{sec_2} introduces the system model of the considered ISAC system and the corresponding performance metrics. Sec. \ref{ACF_Sec} characterizes the ACF of pulse-shaped random ISAC signals. Sec. \ref{iceberg_inspired_design} discusses the design guidelines for ISAC transmission inspired by the iceberg metaphor. Sec. \ref{sec_5} provides simulation results to validate the theoretical analysis of the paper. Finally, Sec. \ref{sec_6} concludes the paper.
\\\indent {\emph{Notations}}: Matrices are denoted by bold uppercase letters (e.g., $\mathbf{U}$), vectors are represented by bold lowercase letters (e.g., $\mathbf{x}$), and scalars are denoted by normal font (e.g., $N$); The $n$th entry of a vector $\mathbf{s}$, and the $(m,n)$-th entry of a matrix $\mathbf{A}$ are denoted as $s_n$ and $a_{m.n}$, respectively; $\otimes$, $\odot$ and $\operatorname{vec}\left(\cdot\right)$ denote the Kronecker product, Hadamard product, and vectorization; $\left(\cdot\right)^T$, $\left(\cdot\right)^H$, and $\left(\cdot\right)^\ast$ stand for transpose, Hermitian transpose, and complex conjugate of the matrices; $\operatorname{Re}\left(\cdot\right)$ and $\operatorname{Im}\left(\cdot\right)$ denote the real and imaginary parts of the argument; $\ell_p$ norm and the Frobenius norm are written as $\left\| \cdot\right\|_p$ and $\left\| \cdot\right\|_F$; $\mathbb{E}(\cdot)$ and $\operatorname{var}(\cdot)$ represent the expectation and variance of the random variable, respectively. 

\section{System Model}\label{sec_2}

Let us examine a single-antenna monostatic ISAC system. The ISAC transmitter (Tx) sends out a signal that is encoded with random communication symbols. This transmitted ISAC signal is captured by a communication receiver (Rx) and is also echoed back by various targets at differing distances to the sensing Rx. Positioned together with the ISAC Tx, the sensing receiver employs matched-filtering (MF) to measure the delay parameters of these targets, relying on the predetermined random ISAC signal.
\subsection{Constellation Model}
Define the vector $\mathbf{s} = \left[s_1, s_2,\ldots,s_{N}\right]^T\in \mathbb{C}^{N\times 1}$ as $N$ transmitted communication symbols that are independently and identically (i.i.d) drawn from a complex constellation $\mathcal{S}$. For ease of analysis, several fundamental assumptions about these constellations are made below, applicable universally.
\begin{assumption}[Unit Power]
    We focus on normalized constellations with a unit power, namely, 
    \begin{equation}
        \mathbb{E}(|s|^2) = 1,\quad \forall s\in\mathcal{S}.
    \end{equation}
\end{assumption}
\noindent This encapsulates PSK but also QAM and other constellations with appropriate scaling.
\begin{assumption}[Rotational Symmetry]
    The expectation and pseudo variance of the constellation are zero, namely
    \begin{equation}
        \mathbb{E}(s) = 0,\quad\mathbb{E}(s^2) = 0, \quad \forall s\in\mathcal{S}.
    \end{equation}
\end{assumption}
\noindent We remark that most of the commonly employed constellations meet the criterion in Assumption 2, including all the PSK and QAM constellations except for BPSK and 8-QAM.

Let us further define
\begin{equation}
    \mu_4 \triangleq \frac{\mathbb{E}\left\{|s-\mathbb{E}(s)|^4\right\}}{\mathbb{E}\left\{|s-\mathbb{E}(s)|^2\right\}^2} =  \mathbb{E}(|s|^4),\quad \forall s\in\mathcal{S}.
\end{equation}
which is known as the \textit{kurtosis} of the constellation. 

The standard complex Gaussian distribution adheres to the criteria outlined above. Known for maximizing channel capacity in Gaussian channels, signals drawn from a Gaussian distribution, with a kurtosis of 2, also serves as a key benchmark for assessing the sensing capabilities of the ISAC signal. Consequently, we introduce and classify two types of constellations below.
\begin{definition}[Sub-Gaussian Constellation]
    A sub-Gaussian constellation is a constellation with kurtosis less than 2, that adheres to Assumptions 1 and 2.
\end{definition}
\begin{definition}[Super-Gaussian Constellation]
    A super-Gaussian constellation is a constellation with kurtosis greater than 2, that adheres to Assumptions 1 and 2.
\end{definition}

QAM and PSK constellations fall under the sub-Gaussian category. Notably, PSK constellations exhibit a kurtosis of 1, whereas QAM constellations feature kurtosis values ranging from 1 to 2. The kurtosis values for standard QAM and PSK constellations are displayed in TABLE. \ref{tab: kurtosis}. On the other hand, super-Gaussian constellations, which show considerably variable amplitudes, can be developed using either geometric or probabilistic methods for constellation shaping. These constellations are especially beneficial in contexts where energy efficiency is paramount or where non-coherent communication methods are employed \cite{923716,1532206,1532207}.

\begin{table}[!t]
\caption{Kurtosis values of typical sub-Gaussian constellations}
\label{tab: kurtosis}
\begin{tabular}{l|c|c|c|c}
\hline
\textbf{Constellation} & PSK     & 16-QAM  & 64-QAM   & 128-QAM  \\ \hline
\textbf{Kurtosis}      & 1       & 1.32    & 1.381   & 1.3427   \\ \hline
\textbf{Constellation} & 256-QAM & 512-QAM & 1024-QAM & 2048-QAM \\ \hline
\textbf{Kurtosis}      & 1.3953  & 1.3506  & 1.3988   & 1.3525   \\ \hline
\end{tabular}
\end{table}

\subsection{Modulation Basis}
In standard communication systems, we modulate the symbol vector $\mathbf{s}$ over an orthonormal basis on the time domain, which is sometimes referred to as a ``waveform". This can be characterized as a unitary matrix $\mathbf{U} = \left[\mathbf{u}_1,\mathbf{u}_2,\ldots,\mathbf{u}_{N}\right]\in \mathbb{U}\left(N\right)$, where $\mathbb{U}\left(N\right):\left\{\mathbf{U}\in\mathbb{C}^{N \times N}| \mathbf{U}\mathbf{U}^H = \mathbf{U}^H\mathbf{U} = \mathbf{I}_N\right\}$ stands for the $N$-dimensional unitary group.
The discrete time-domain signal sample $\mathbf{x} = \left[x_1, x_2,\ldots,x_{N}\right]^T\in \mathbb{C}^{N\times 1}$ can then be expressed as
\begin{equation}
    \mathbf{x} = \mathbf{U}\mathbf{s} = \sum\limits_{n = 1}^{N} {s_n\mathbf{u}_n}.
\end{equation}
The above generic model may represent most of the orthogonal communication signaling schemes, such as SC, OFDM, CDMA, OTFS, and AFDM. We refer the readers to \cite{liu2024ofdm} for more detailed examples. In these schemes, the addition of a cyclic prefix (CP) is often necessary, which eliminates the inter-symbol interference (ISI) caused by multi-path effect, and reduces the computational complexity by processing the received signal in the frequency/delay-Doppler/code/affine domains. Towards that end, we consider ISAC signaling with CP, which correspond to the periodic convolution processing of the MF at the sensing Rx. Without loss of generality, the CP is assumed to be larger than the maximum delay of the communication paths and sensing targets.

\subsection{Pulse Shaping Model}
In order to transmit the signal $\mathbf{x}$ in the band-limited ISAC channel, pulse shaping is indispensable as a means to restrict the bandwidth of the signal and to eliminate the ISI among the time-domain samples. Let $p(t)$ be a band-limited prototype Nyquist pulse with an one-sided bandwidth $B$ and a roll-off factor $\alpha$. The continuous pulse-shaped signal may then be expressed as
\begin{equation}
    \tilde{x}(t) = \sum\limits_{n = 1}^{N}x_n p(t-nT),
\end{equation}
where $T = \frac{1+\alpha}{2B}$ is the symbol duration. 
Upon relying on the unit impulse function $\delta(t)$, and by adding CP to $\mathbf{x}$, $\tilde{x}(t)$ can be alternatively expressed as the following circular convolution:
\begin{equation}
    \tilde{x}(t) = \sum\limits_{n = 1}^{N}x_n \delta(t-nT)\circledast p(t).
\end{equation}


We proceed our study using an oversampling based implementation with a sampling rate $f_s = \frac{1}{T_s}$, with $T_s$ being the sampling duration. 
Without loss of generality, we assume that the over-sampling ratio $L = \frac{T}{T_s} \ge 1$ is an integer, where $L = 1$ corresponds to the Nyquist sampling. By doing so, the $k$th sample of the pulse-shaped signal $x(t)$ is given by
\begin{align}\label{pulse_shaped_signal}
   \tilde{x}_{k} = \sum\limits_{n = 1}^{N}x_n \delta(kT_s-nT)\circledast p(kT_s),\;k = 0,1,\ldots, LN-1, 
\end{align}
Let $p_k \triangleq p(kT_s)$, and denote $\mathbf{p} = \left[p_0,p_1,\ldots,p_{LN-1}\right]^T$, with its energy being normalized to $\left\|\mathbf{p}\right\|^2 = 1$. By defining the up-sampled signal as
\begin{align}
\mathbf{x}_{\rm up} = \left[x_1,\mathbf{0}_{L-1}^T, x_2,\mathbf{0}_{L-1}^T,\ldots,x_N,\mathbf{0}_{L-1}^T\right]^T,
\end{align}
we may recast \eqref{pulse_shaped_signal} into the following compact form:
\begin{equation}
\tilde{\mathbf{x}} = \mathbf{P}\mathbf{x}_{\rm up},
\end{equation}
where $\mathbf{P}\in \mathbb{C}^{LN \times LN}$ is a circulant matrix, given by
\begin{equation}
{\mathbf{P}} = \left[ {\begin{array}{*{20}{c}}
  {{p_0}}&{{p_{LN-1}}}& \ldots &{{p_1}} \\ 
  {{p_1}}&{{p_0}}& \ldots &{{p_2}} \\ 
   \vdots & \vdots & \ddots & \vdots  \\ 
  {{p_{LN-1}}}&{{p_{LN-2}}}& \ldots &{{p_0}} 
\end{array}} \right].
\end{equation}
In practice, the resultant sequence $\tilde{\mathbf{x}}$ will be passed to a digital-to-analog converter (DAC) for approximating the continuous waveform $\tilde{x}(t)$ with a predefined chip pulse.
For the ease of study, we henceforth discuss the ranging problem in a discrete form by considering the sampled sequence $\tilde{\mathbf{x}}$. This yields a practical and good approximation of the actual problem in hand, especially when the over-sampling ratio is high.

\subsection{Sensing Signal Processing based on Matched Filtering}
Let us consider the scenario where $Q$ targets located at different ranges need to be sensed simultaneously. Without the loss of generality, suppose that the sampling ratio $L$ is sufficiently large such that all targets are on grid. The received echo signal after sampling is given by
\begin{equation}
    \mathbf{y} = \sum\limits_{q = 1}^Q\alpha_q\mathbf{J}_{\tau_q}\tilde{\mathbf{x}} + \mathbf{z},
\end{equation}
where $\mathbf{z}\sim\mathcal{CN}\left({\mathbf{0},\sigma_z^2\mathbf{I}}\right)$ stands for the white Gaussian noise, $\alpha_q$ and $\tau_q$ are the complex reflection coefficient and delay (normalized by $T_s$) of the $q$th target, respectively, and $\mathbf{J}_k\in\mathbb{R}^{LN\times LN}$ represents the $k$th periodic time-shift matrix due to the addition of the CP, which is
\begin{equation}
    \mathbf{J}_k = \left[ {\begin{array}{*{20}{c}}
  {\mathbf{0}}&{{{\mathbf{I}}_{LN - k}}} \\ 
  {\mathbf{I}_k}&{\mathbf{0}} 
\end{array}} \right],
\end{equation}
and
\begin{equation}
    \mathbf{J}_{-k} = \mathbf{J}_{LN-k} = \mathbf{J}_k^T= \left[ {\begin{array}{*{20}{c}}
  {\mathbf{0}}&{{{\mathbf{I}}_{k}}} \\ 
  {\mathbf{I}_{LN - k}}&{\mathbf{0}} 
\end{array}} \right].
\end{equation}
Accordingly, the auto-correlation function (ACF) of $\tilde{\mathbf{x}}$ can be defined as
\begin{equation}
    R_k = \tilde{\mathbf{x}}^H{{\mathbf{J}}_{{k}}}\tilde{\mathbf{x}},\;\; k = 0,1,\ldots,LN-1.
\end{equation}
The ACF of the ISAC signal is an important performance indicator for ranging tasks. In particular, the mainlobe width of the ACF determines the range resolution, which is usually inversely proportional to the signal bandwidth. The sidelobe level, on the other hand, is critical for multi-target detection.

To extract the delay parameters, a common practice is to matched filter (MF) the echo signal $\mathbf{y}$ with the transmitted signal $\tilde{\mathbf{x}}$, yielding the following MF output:
\begin{align}
    {{\tilde y}_i} &\nonumber = \tilde{\mathbf{x}}^H{{\mathbf{J}}_{i}^T}{\mathbf{y}} = \sum\limits_{q = 1}^Q {{\alpha _q}} \tilde{\mathbf{x}}^H{{\mathbf{J}}_{ i}^T}{{\mathbf{J}}_{{\tau _q}}}\tilde{\mathbf{x}} + \tilde{\mathbf{x}}^H{{\mathbf{J}}_{i}^T}{\mathbf{z}} \\
    &\nonumber = \sum\limits_{q = 1}^Q {{\alpha _q}} \tilde{\mathbf{x}}^H{{\mathbf{J}}_{{\tau _q-i}}}\tilde{\mathbf{x}} + \tilde{\mathbf{x}}^H{{\mathbf{J}}_{i}^T}{\mathbf{z}},\\
    & = \sum\limits_{q = 1}^Q {{\alpha _q}} R_{\tau_q-i} + \tilde{z}_i,\;\; i = 0,1,\ldots,LN-1,
\end{align}
which may be viewed as a linear combination of time-shifted ACFs plus noise. In order to facilitate the detection of targets, it is desired that the squared output $|{{\tilde y}_i}|^2$ generates high peaks at $i = \tau_q$, and small sidelobe levels elsewhere, which depends heavily on the overall geometry of the ACF. Note that the ACF is a random function due to the randomness of communication symbols, in which case we have to evaluate its statistical properties rather than a specific realization. In particular, we are interested in the expectation of the squared ACF, namely,
\begin{equation}\label{squared_ACF}
    \mathbb{E}(|R_k|^2) = \mathbb{E}(|\tilde{\mathbf{x}}^H{{\mathbf{J}}_{{k}}}\tilde{\mathbf{x}}|^2), \;\; k = 0,1,\ldots,LN-1,
\end{equation}
where the expectation is over $\mathbf{s}$. Here the square is imposed to measure the average mainlobe and sidelobe levels of the ACF, as $R_k$ is a complex function.

In the next section, we aim to derive a closed-form expression of \eqref{squared_ACF} under arbitrary orthogonal waveform $\mathbf{U}$ and Nyquist pulse $\mathbf{p}$, for i.i.d. symbol sequences drawn from any proper constellation $\mathcal{S}$. We show that the structure of \eqref{squared_ACF} may be understood as an ``iceberg'' partially hidden in the ``sea'', which are resultant from the ACF of the pulse and the randomness of the communication data, respectively. This offers important insights to the design of the waveform, pulse shaping, and constellation design for ISAC systems.

\section{The ACF of Pulse-Shaped Random ISAC Signals}\label{ACF_Sec}
\subsection{Frequency-Domain Representation of the ACF}
We commence by showing the basic structure of the circulant matrix in the following lemma.
\begin{lemma}
    A size-$N$ circulant matrix $\mathbf{C}$ can be diagonalized by the DFT matrix in the form of
    \begin{equation}
        \mathbf{C} = \sqrt{N}\mathbf{F}_N^H\operatorname{Diag}(\mathbf{F}_N\mathbf{c})\mathbf{F}_{N},
    \end{equation}
    where $\mathbf{c}$ is the first column of $\mathbf{C}$, and $\mathbf{F}_N$ is the normalized DFT matrix of size $N$, with its $(m,n)$-th entry being defined as $\frac{1}{\sqrt{N}}e^{-\frac{j2\pi(m-1)(n-1)}{N}}$.
\end{lemma}
Since both $\mathbf{J}_k$ and $\mathbf{P}$ are circulant matrices, we have
\begin{align}
    &{\mathbf{J}}_k  =  \sqrt{LN}\mathbf{F}_{LN}^H\operatorname{Diag}(\mathbf{f}_{LN-k+1})\mathbf{F}_{LN},\\
    & \mathbf{P} = \sqrt{LN}\mathbf{F}_{LN}^H\operatorname{Diag}(\mathbf{F}_{LN}\mathbf{p})\mathbf{F}_{LN},
\end{align}
where $\mathbf{f}_{n}$ is the $n$th column of $\mathbf{F}_{LN}$. The ACF can then be alternatively expressed as
\begin{align}
    R_k  &\nonumber = \tilde{\mathbf{x}}^H{{\mathbf{J}}_{{k}}}\tilde{\mathbf{x}} = \sqrt{LN}\tilde{\mathbf{x}}^H\mathbf{F}_N^H\operatorname{Diag}(\mathbf{f}_{LN-k+1})\mathbf{F}_{LN}\tilde{\mathbf{x}}\\
    &\nonumber = \mathbf{x}_{\rm up}^H\mathbf{P}^H\mathbf{F}_N^H\operatorname{Diag}(\mathbf{f}_{LN-k+1})\mathbf{F}_{LN}\mathbf{P}\mathbf{x}_{\rm up}\\
    &\nonumber = (LN)^{\frac{3}{2}}\mathbf{x}_{\rm up}^H\mathbf{F}_{LN}^H\operatorname{Diag}(\mathbf{F}_{LN}\mathbf{p}\odot\mathbf{f}_{LN-k+1}\odot \mathbf{F}_{LN}^\ast\mathbf{p}^\ast)\\
    &\quad\;\;\cdot\mathbf{F}_{LN}\mathbf{x}_{\rm up}.
\end{align}
The DFT of the up-sampling signal $\mathbf{x}_{\rm up}$ is known to be the periodic extension of the spectrum of the original signal, namely,
\begin{align}\label{periodic_extension}
    \mathbf{F}_{LN}\mathbf{x}_{\rm up} &\nonumber = \sqrt{1/LN}\left[\mathbf{x}^T\mathbf{F}_N^T,\mathbf{x}^T\mathbf{F}_N^T,\ldots,\mathbf{x}^T\mathbf{F}_N^T\right]^T \\ &=\sqrt{1/LN}\left[\mathbf{s}^T\mathbf{V}^\ast,\mathbf{s}^T\mathbf{V}^\ast,\ldots,\mathbf{s}^T\mathbf{V}^\ast\right]^T\in \mathbb{C}^{LN \times 1},
\end{align}
where $\mathbf{V} = \left[\mathbf{v}_1,\mathbf{v}_2,\ldots,\mathbf{v}_N\right]\in \mathbb{C}^{N \times N}$ is defined as $\mathbf{V} = \mathbf{U}^H\mathbf{F}_N^H$, such that $\mathbf{F}_N\mathbf{x} = \mathbf{V}^H\mathbf{s}$. 

Let $\mathbf{g} = \left[g_1,g_2,\ldots,g_{LN}\right]^T$ be the squared spectrum of $\sqrt{N}\mathbf{p}$, i.e., 
\begin{equation}
    \mathbf{g} = N\mathbf{F}_{LN}\mathbf{p}\odot\mathbf{F}_{LN}^\ast\mathbf{p}^\ast.
\end{equation}
Note that $\sum\nolimits_{n = 1}^{LN}g_n = N$ due to the normalized pulse $\mathbf{p}$. To eliminate the ISI, Nyquist pulse shaping with a roll-off factor $0\le\alpha\le1$ is adopted, occupying only the first and last periods of \eqref{periodic_extension}. With this observation, and by noting the fact that $\mathbf{f}_{LN-k+1} = \mathbf{f}_{k+1}^\ast$, we obtain a frequency-domain representation of the ACF as
\begin{equation}
R_k = \sum\limits_{n = 1}^{N}g_n|\mathbf{v}_n^H\mathbf{s}|^2 e^{\frac{j2\pi k(n-1)}{LN}} + g_{(L-1)N+n}|\mathbf{v}_n^H\mathbf{s}|^2 e^{\frac{j2\pi k(n-N-1)}{LN}}.
\end{equation}
where $g_n$ are set to zero for $n\in \left[N+1,(L-1)N\right]$ due to $\alpha \le 1$. The frequency-domain condition of the Nyquist pulse, known as the \textit{folded spectrum criterion}\cite{proakis2008digital}, requires that
\begin{equation}
    g_{(L-1)N+n} = 1-g_n, \quad n = 1,2,\ldots,N,
\end{equation}
leading to 
\begin{equation}\label{compact_expression}
    R_k = \sum\limits_{n = 1}^{N}\tilde{g}_{n,k}|\mathbf{v}_n^H\mathbf{s}|^2 e^{\frac{j2\pi k(n-1)}{LN}},
\end{equation}
where $\tilde{g}_{n,k} = g_n + (1-g_n)e^{-\frac{j2\pi k}{L}}$.

\subsection{Characterization of the Average Squared ACF}
With the compact expression \eqref{compact_expression} at hand, we may express \eqref{squared_ACF} in an analytic form. 

\begin{thm}[Iceberg Theorem]
    The average squared ACF is
    \begin{small}
    \begin{align}\label{squared_P-ACF}
        \mathbb{E}(|{R_k}{|^2}) &\nonumber = \underbrace {N{{\left| {\tilde{\mathbf{f}}_{k + 1}^H\tilde {\mathbf{g}}_k} \right|}^2}}_{\text{Iceberg}} + \underbrace {\left\| {\tilde {\mathbf{g}}_k} \right\|^2 + ({\mu _4} - 2)N\left\| {\tilde {\mathbf{V}}\left( {\tilde {\mathbf{g}}_k \odot \tilde{\mathbf{f}}_{k + 1}^\ast} \right)} \right\|^2}_{\text{Sea Level}}\\
        & = \left|\mathbb{E}({R_k})\right|^2 + \operatorname{var}(R_k),\quad k = 0,1,\ldots,LN-1,
    \end{align}
    \end{small}
    where $\tilde{\mathbf{f}}_{k + 1}$ contains the first $N$ entries of ${\mathbf{f}}_{k + 1}$, and
    \begin{align}
        \tilde{\mathbf{V}} = \mathbf{V}\odot\mathbf{V}^\ast, \quad \tilde{\mathbf{g}}_k = \left[\tilde{g}_{1,k},\tilde{g}_{2,k},\ldots,\tilde{g}_{N,k}\right]^T.
    \end{align}
    Moreover, the iceberg and sea level are the squared mean and variance of $R_k$, respectively. 
\end{thm}
\begin{proof}
    See Appendix. \ref{Theorem_1_proof}.
\end{proof}

\noindent Observe that the ``iceberg'' part is the squared IDFT of a zero-padded version of $\tilde{\mathbf{g}}_k$, and is also equivalent to the squared ACF of the pulse shaping filter. This equivalence arises because the IDFT of $\tilde{g}_{n,k} = g_n + (1-g_n)e^{-\frac{j2\pi k}{L}}$ can be interpreted as the IDFT of $g_n$ combined with that of a frequency-shifted version of $1-g_n$, which is exactly the IDFT of the squared spectrum of the Nyquist pulse $\mathbf{p}$, yielding its ACF according to the Wiener–Khinchin theorem. 
\begin{corollary}
    Under the Nyquist pulse shaping, the average mainlobe level depends only on the kurtosis of the constellation, which is
    \begin{equation}\label{mainlobe_0}
        \mathbb{E}(|R_0|^2) = N^2 + (\mu_4 - 1)N.
    \end{equation}
\end{corollary}
\begin{proof}
    We have $\tilde{g}_{n,0} = 1$ by its definition. Moreover, note that $\tilde{\mathbf{f}}_1 = \frac{1}{\sqrt{N}}\mathbf{1}_N$. This indicates that
    \begin{equation}\label{mainlobe_1}
        N{\left| {{\mathbf{f}}_{1}^H\tilde {\mathbf{g}}_0} \right|^2} = N^2, \quad \left\| {\tilde {\mathbf{g}}}_0 \right\|^2 = N, \quad {\tilde {\mathbf{g}}_0\odot {\tilde{\mathbf{f}}}_{1}^\ast} = \frac{1}{\sqrt{N}}\mathbf{1}_N.
    \end{equation}
    Furthermore, since $\tilde{\mathbf{V}} = \mathbf{V}\odot\mathbf{V}^\ast$ is a unistochastic matrix\footnote{A unistochastic matrix is generated from the entry-wise square of a unitary matrix, which is also bistochastic.} due to the unitarity of $\mathbf{V}$, we have
    \begin{equation}\label{mainlobe_2}
        \tilde{\mathbf{V}}\mathbf{1}_N = \mathbf{1}_N.
    \end{equation}
    Substituting \eqref{mainlobe_1} and \eqref{mainlobe_2} to \eqref{ERk2} yields \eqref{mainlobe_0}.
\end{proof}
As an illustrative example, Fig. \ref{iceberg_example} compares the average squared ACF of a pulse-shaped SC signal with the squared ACF of the pulse itself. The constellation used here is 16-QAM, the pulse shaping filter is the well-known root-raised cosine (RRC) filter with a roll-off factor of $\alpha = 0.35$, and the number of symbols and oversampling ratio are set to $N = 128$ and $L = 10$, respectively. The results show that the ACF of the random SC signal closely matches the ACF of the pulse within the delay region $\left[-2, 2\right]$, corresponding to the ``tip'' of the iceberg. Beyond this region, the sidelobe level is primarily attributed to the ``sea level'' component. 

\noindent\textbf{Remark 1} (Iceberg Metaphor): Theorem 1 reveals that the average squared ACF consists of two components: the ``iceberg'' (squared mean of $R_k$) and the ``sea level'' (variance of $R_k$). Specifically, the ``iceberg'' represents the squared ACF of the time-domain pulse. This part depends on the pulse shaping, and determines the overall shape of the ACF, which primarily impacts the sensing performance for targets near the mainlobe. In contrast, the ``sea level'' arises from the randomness of the data payload, and affects the sensing performance for targets further from the mainlobe. In the subsequent subsection, we will demonstrate that the ``sea level'' can be effectively reduced through the coherent integration technique.

\subsection{Coherent Integration}
To further enhance the sensing performance in practical scenarios, a coherent integration technique can be employed. Assuming that the targets remain staitionary across $M$ transmission slots, one may randomly realize $M$ i.i.d. symbol sequences $\mathbf{s}_1, \mathbf{s}_2,\ldots,\mathbf{s}_M$ from certain constellation, and average over their corresponding MF output $\tilde{y}_i^{(m)}$, given by
\begin{equation}\label{CI_MF}
    \frac{1}{M}\sum\limits_{m = 1}^{M}\tilde{y}_i^{(m)} = \frac{1}{M}\sum\limits_{q = 1}^{Q}{{\alpha _q}}\sum\limits_{m = 1}^M R_{\tau_q-i}^{(m)} + \tilde{z}_i^{(m)},
\end{equation}
where $R_k^{(m)}$ represents the ACF of the $m$th realization of the ISAC signal, and $\tilde{z}_i^{(m)}$ is the noise. To evaluate the sensing performance, \eqref{CI_MF} requires an analysis of the coherently integrated ACF of random ISAC signals, which is
\begin{equation}
  \overline{R}_k = \frac{1}{M}\sum\limits_{m = 1}^{M}R_k^{(m)} = \frac{1}{M}\sum\limits_{m = 1}^{M}\sum\limits_{n = 1}^{N}\tilde{g}_{n,k}|\mathbf{v}_n^H\mathbf{s}_m|^2 e^{\frac{j2\pi k(n-1)}{LN}}.
\end{equation}
Note that $\overline{R}_k$ is still a complex-valued function. To assess its mainlobe and sidelobe levels, we analyze again the expectation of its squared value in the following corollary. 

\begin{corollary}[Sea Level Reduction via Coherent Integration]
Through coherently integrating $R_k$ over $M$ transmission slots, the average squared ACF becomes
    \begin{align}\label{CIed_squared_P-ACF}
        &\nonumber\mathbb{E}(|\overline{R}_k|^2)  = \underbrace {N{{\left| {\tilde{\mathbf{f}}_{k + 1}^H\tilde {\mathbf{g}}_k} \right|}^2}}_{\text{Iceberg}} \\
        & + \underbrace {\frac{1}{M}\left\{\left\| {\tilde {\mathbf{g}}_k} \right\|^2 + ({\mu _4} - 2)N\left\| {\tilde {\mathbf{V}}\left( {\tilde {\mathbf{g}}_k\odot \tilde{\mathbf{f}}_{k + 1}^\ast} \right)} \right\|^2\right\}}_{\text{Sea Level}}.
    \end{align}
\end{corollary}
\begin{proof}
    Observe the fact that $\overline{R}_k$ is nothing but the sample mean of $R_k$ averaged over $M$ i.i.d. instances. It is straightforward that the expectation keeps unchanged, while the variance, namely, the sea level of $\overline{R}_k$ is reduced by a factor of $M$. 
\end{proof}
\begin{corollary}
    The mainlobe level after coherent integration by $M$ times becomes
    \begin{equation}\label{mainlobe_CIed}
        \mathbb{E}(|\overline{R}_0|^2) = N^2 + \frac{(\mu_4 - 1)N}{M}.
    \end{equation}
\end{corollary}

\begin{figure}[!t]
	\centering
	\includegraphics[width = \columnwidth]{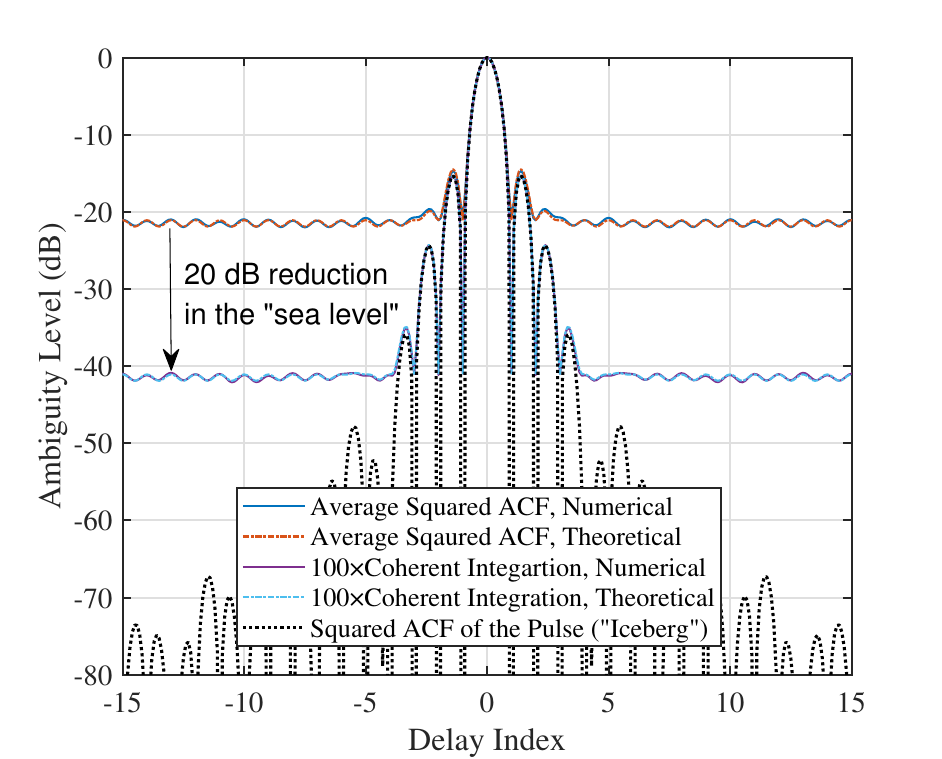}
	\caption{The average squared ACF and its coherent integration version of an SC signal, with 16-QAM constellation and $\alpha = 0.35$ RRC pulse shaping, $N = 128$, $L = 10$, $M = 100$.}
    \label{iceberg_example}
\end{figure} 

\noindent\textbf{Remark 2}: Corollary 2 demonstrates that coherent integration can reduce the ``sea level'' component of the average squared ACF by a factor of $M$. This effect is also illustrated in Fig. \ref{iceberg_example}, where the ACF obtained from $M = 100$ coherent integrations is compared with its no-integration counterpart, with all other parameters remaining consistent. As predicted by Corollary 2, a 20 dB reduction in the ``sea level'' is observed, unveiling more of the ``iceberg'' within the delay region $[-4, 4]$. From Corollary 2, it is also evident that the ranging performance limit of random ISAC signals relies on the chosen pulse shaping filter, as $\mathbb{E}(|\overline{R}_k|^2)$ approaches the ``iceberg'' when $M\to\infty$.

\begin{figure*}[htbp]
\normalsize
\newcounter{MYtempeqncnt3}
\setcounter{MYtempeqncnt3}{\value{equation}}
\setcounter{equation}{34}
\begin{align}
    &\mathbb{E}\left( {{{\left| {{{\overline{R}}_k^{\text{OFDM}}}} \right|}^2}} \right) = \underbrace {{{\left| {\sum\limits_{n = 1}^N {\left( {{g_n} + \left( {1 - {g_n}} \right){e^{\frac{{j2\pi k}}{L}}}} \right){e^{\frac{{j2\pi \left( {n - 1} \right)k}}{{LN}}}}} } \right|}^2}}_{\text{Iceberg}} + \underbrace {\frac{{\mu _4} - 1}{M}\left\{{N - 2\left({1 - \cos \frac{{2\pi k}}{L}} \right)\sum\limits_{n = 1}^N {{g_n}\left( {1 - {g_n}} \right)} } \right\}}_{\text{Sea Level}}. \label{OFDM_Iceberg}\\
    &\mathbb{E}\left( {{{\left| {{{\overline{R}}_k^{\text{SC}}}} \right|}^2}} \right)  = \left( {1 + \frac{{{\mu _4} - 2}}{MN}} \right)\underbrace {{{\left| {\sum\limits_{n = 1}^N {\left( {{g_n} + \left( {1 - {g_n}} \right){e^{\frac{{j2\pi k}}{L}}}} \right){e^{\frac{{j2\pi \left( {n - 1} \right)k}}{{LN}}}}} } \right|}^2}}_{\text{Iceberg}} + \underbrace {\frac{1}{M}\left\{N - 2\left( {1 - \cos \frac{{2\pi k}}{L}} \right)\sum\limits_{n = 1}^N {{g_n}\left( {1 - {g_n}} \right)}\right\} }_{\text{Sea Level}}. \label{SC_Iceberg}
\end{align}
\setcounter{equation}{\value{MYtempeqncnt3}}
\hrulefill
\vspace*{4pt}
\end{figure*}

\section{Iceberg Theorem Inspired ISAC Transmission Design}\label{iceberg_inspired_design}
The findings presented in Sec. \ref{ACF_Sec} prompt a reconsideration of ISAC transmission design when employing random signaling. Insights from Theorem 1 and Corollary 2 reveal that the average squared ACF of random ISAC signals is fundamentally influenced by the choice of modulation basis, constellation, and pulse shaping filter. We elaborate on the impact of each of these components in the following discussion.
\subsection{Optimal Modulation Basis}
Corollaries 1 and 2 indicate that the mainlobe level remains constant regardless of the chosen signaling basis when Nyquist pulse shaping is applied. To enhance sensing performance, it is crucial to design a modulation basis that reduces the sidelobe level. Notably, the modulation basis primarily influences the ``sea level'' component of the sidelobes in both \eqref{squared_P-ACF} and \eqref{CIed_squared_P-ACF}. The optimal modulation basis design thus depends on the sign of $\mu_4 - 2$, commonly referred to as the \textit{excess kurtosis}.

We first discuss the optimal modulation basis for sub-Gaussian ($\mu_4<2$) constellations, e.g., QAM and PSK, by proving the following theorem. 
\begin{thm}
    For sub-Gaussian constellations, OFDM is the only modulation basis that achieves the lowest ranging sidelobe level at every lag $k$.
\end{thm}
\begin{proof}
See Appendix. \ref{Theorem_2_Proof}.
\end{proof}
\noindent Note that the uniqueness of OFDM as the optimal modulation basis for ranging is guaranteed, as being detailed in the proof.

The following theorem provides the optimal modulation basis for super-Gaussian ($\mu_4>2$) constellations.
\begin{thm}
    For super-Gaussian constellations, SC modulation achieves the lowest ranging sidelobe level for every $k$.
\end{thm}
\begin{proof}
See Appendix. \ref{Theorem_3_proof}.
\end{proof}
\noindent \textbf{Remark 3}: Theorems 2 and 3 offer valuable insights into modulation schemes for ISAC signals, indicating that the optimal modulation basis is determined by the sign of the excess kurtosis of the constellation. Considering that sub-Gaussian constellations are prevalent in modern communication systems, OFDM emerges as the optimal choice for ranging tasks using MF methods under random signaling. Notably, when a Gaussian constellation is employed, i.e., $\mathbf{s} \sim \mathcal{CN}(\mathbf{0}, \mathbf{I}_N)$, the sidelobe becomes independent of the modulation basis, as the Gaussian distribution is unitary invariant.

By substituting $\mathbf{U} = \mathbf{F}_N^H$ and $\mathbf{U} = \mathbf{I}_N$ into \eqref{CIed_squared_P-ACF}, we may obtain the average squared ACFs under OFDM and SC modulations in \eqref{OFDM_Iceberg} and \eqref{SC_Iceberg} at the top of this page.
\subsection{Constellation Design}
The average squared ACF is affected by the constellation exclusively through its kurtosis. As discussed earlier, the sign of the excess kurtosis dictates the optimal choice of modulation basis. For a given pair of modulation basis and pulse shaping filter, it follows from \eqref{squared_P-ACF} that constellations with lower kurtosis result in a reduced ``sea level''. Notably, PSK constellations, with $\mu_4 = 1$, achieve the lowest possible ``sea level''. An interesting observation arises from \eqref{OFDM_Iceberg}: when using an OFDM modulation, PSK can completely ``drain the sea'' by reducing the ``sea level'' component to zero. In this case, the auto-correlation properties of the random ISAC signal are determined solely by the pulse shaping, represented as the ``iceberg'' in \eqref{squared_P-ACF}.

Let us further investigate the influence of the constellation's kurtosis from the following two aspects:
\subsubsection{Mainlobe Level}
The mainlobe level in no-integration and coherent integration cases are shown in \eqref{mainlobe_0} and \eqref{mainlobe_CIed}, respectively, which are both monotonically increasing functions of $\mu_4$. This suggests that constellations with larger kurtosis yield higher mainlobe. However, when $N$ is sufficiently large, the mainlobe level can generally be approximated as $N^2$ because $\mu_4 \ll N$, making the contribution from kurtosis negligible. This effect is particularly pronounced in the coherent integration case, where the kurtosis contribution is further reduced by a factor of $M$.
\subsubsection{Integration Efficiency}
The constellation kurtosis also affects the performance of coherent integration, particularly the efficiency in reducing the ``sea level''. Let us take the OFDM signaling as an example. According to \eqref{OFDM_Iceberg}, the sea level is proportional to $(\mu_4 - 1)/M$, where the integration efficiency can be naturally defined as $1/(\mu_4 - 1)$. For a 16-QAM constellation with $\mu_4 = 1.32$, the integration efficiency is 3.125. In contrast, a Gaussian constellation with $\mu_4 = 2$ has an integration efficiency of 1. This means that 16-QAM is more than 3 times as efficient as a Gaussian constellation. 

From the above discussion, it is evident that reducing the kurtosis of constellations is essential for improving the sensing performance. However, constellations with low kurtosis may not always support high communication rates. For example, while PSK has the smallest kurtosis, it generally achieves lower communication rates compared to QAM of the same order. To achieve a scalable tradeoff between sensing and communication, recent research has explored probabilistic constellation shaping (PCS) to optimize the input distribution of OFDM signals, thereby maximizing the achievable communication rate under a given constellation kurtosis. Readers are referred to \cite{du2024reshaping} for more technical details.

\noindent\textbf{Remark 4}: It is important to note that constellation shaping may not always be effective, particularly with SC signaling. As observed from \eqref{SC_Iceberg}, reducing kurtosis has only a minor impact on its ``iceberg'' component. This is because when $\mathbf{V} = \mathbf{F}_N^H$, we have 
\setcounter{equation}{36}
\begin{equation}
   \tilde{\mathbf{V}} = \frac{1}{N}\mathbf{1}_{N}\mathbf{1}_{N}^T,\quad  \left\| {\tilde {\mathbf{V}}\left( {\tilde {\mathbf{g}}_k \odot \tilde{\mathbf{f}}_{k + 1}^\ast} \right)} \right\|^2 = \left| {\tilde{\mathbf{f}}_{k + 1}^H\tilde {\mathbf{g}}_k} \right|^2,
\end{equation}
indicating that a small portion of the ``sea level'' retains a similar shape of the ``iceberg''.  In fact, under the same sub-Gaussian constellation, SC achieves an efficiency of (almost) 1, while OFDM always benefits from a integration efficiency of $1/(\mu_4 - 1) \ge 1$.

\subsection{Pulse Shaping Design}\label{iceberg_shaping}
Finally, we examine the pulse shaping design for random ISAC signals. 
\subsubsection{Mainlobe Width}
Due to the Nyquist property of the pulse, the ``iceberg part'' (namely, the ACF of the pulse) equals to zero when $k = mL$ where $m$ is any non-zero integer. This can be readily proved by noting that
\begin{align}
     &\nonumber{\sum\limits_{n = 1}^N {\left( {{g_n} + \left( {1 - {g_n}} \right){e^{\frac{{j2\pi mL}}{L}}}} \right){e^{\frac{{j2\pi \left( {n - 1} \right)mL}}{{LN}}}}} } \\&= {\sum\limits_{n = 1}^N{e^{\frac{{j2\pi \left( {n - 1} \right)m}}{{N}}}}}
      = N\delta_{0,m},\quad m\in\mathbb{Z}.
\end{align}
We may therefore define the mainlobe width as the distance between two first zeros, which is also the delay resolution, given by
\begin{equation}
    \Delta T = 2LT_s = 2T = \frac{1+\alpha}{B},
\end{equation}
indicating that the delay resolution is proportional to the roll-off factor $\alpha$.

\subsubsection{Sea Waves}
As observed from Fig. \ref{iceberg_example}, periodic ripples appear in the sidelobe region of the average squared ACF, which can be interpreted as ``sea waves'' over the ``sea level". This phenomenon can be explained by \eqref{squared_P-ACF} as follows:
\begin{align}\label{sea_waves}
    \left\| {\tilde {\mathbf{g}}_k} \right\|^2 &\nonumber = \sum\limits_{n = 1}^N\left|g_n+\left({1 - {g_n}} \right){e^{\frac{{j2\pi k}}{L}}}\right|^2\\
    & = N - 2\left(1-\cos{\frac{2\pi k}{L}}\right)\sum\limits_{n = 1}^N g_n(1-g_n),
\end{align}
which indicates that the ripples are generated by the cosine term in \eqref{sea_waves}. The amplitude of these ripples is proportional to the summation $\sum_{n = 1}^N g_n (1 - g_n)$. For Nyquist pulses, where $0 \le g_n \le 1$, we have $\sum_{n = 1}^N g_n (1 - g_n) \ge 0$. Notably, pulses with a larger roll-off factor generally produce higher ripples in the sidelobe region. This occurs because $g_n (1 - g_n) = 0$ when $g_n \in \left\{0, 1\right\}$, meaning that the summation is solely contributed by the roll-off portion where $0 < g_n < 1$. In the extreme case of a sinc pulse, where $g_n$ takes a rectangular shape with equal numbers of zeros and ones, the summation is exactly zero, resulting in a completely flat ``sea level''.

\subsubsection{Iceberg Shaping}
To further boost the sensing performance, it is crucial to carefully design the ``iceberg'' component of the average squared ACF. This is because the coherent integration operation reduces only the ``sea level'' part. After integrating a sufficiently large number of MF outputs, the sensing performance will depend mainly on the geometry of the ``iceberg''. While various techniques exist for designing probing waveforms with favorable auto-correlation properties, we propose a novel approach for designing Nyquist pulses by focusing on minimizing the sidelobe level of its ACF within a specified delay region.

Suppose that we are interested in detecting targets over a delay region $k \in \mathcal{K}_{sl}$. To that end, it is necessary to minimize the sidelobe level within that region. We may adopt either the integrated sidelobe level (ISL) or peak sidelobe level (PSL) as objective functions, defined as
\begin{align}
    & f_{\text{ISL}}= \sum\limits_{k \in \mathcal{K}_{sl}}{{\left| {\tilde{\mathbf{f}}_{k + 1}^H\tilde {\mathbf{g}}_k} \right|}^2}, \\
    & f_{\text{PSL}}= \mathop{\max}\limits_{k \in \mathcal{K}_{sl}} {{\left| {\tilde{\mathbf{f}}_{k + 1}^H\tilde {\mathbf{g}}_k} \right|}^2}, \label{PSL}
\end{align}
both of which are convex in $g_n$. 

Next, we discuss the constraints on $g_n$. Given a roll-off factor $\alpha$, and assuming without loss of generality that $N_{\alpha} = \alpha N$ is an integer and that $N - N_{\alpha}$ is an even number, the non-roll-off portion of $g_n$ will contain $\frac{N-N_{\alpha}}{2}$ zeros and $\frac{N-N_{\alpha}}{2}$ ones, as follows:
\begin{align}
   & g_n = 0, \quad n = 1,2,\ldots,\frac{N-N_{\alpha}}{2},\\
   & g_n = 1, \quad n = \frac{N+N_{\alpha}}{2}+1,\frac{N+N_{\alpha}}{2}+2,\ldots, N.
\end{align}
Moreover, to guarantee that the roll-off part of the spectrum is monotonically increasing, one has to ensure
\begin{equation}
    g_{n+1} - g_{n} \ge 0, n = 1,2,\ldots, N-1.
\end{equation}
Finally, maintaining a constant energy of the pulse requires
\begin{equation}
    \sum\limits_{n = 1}^N g_n = \frac{N}{2}.
\end{equation}
Note that these constraints are all linear in $g_n$. The iceberg shaping problem can be then formulated as
\begin{equation}
\mathop {\min }\limits_{{g_n} \ge 0} \;{f_{{\text{ISL}}}}\;{\text{or}}\;{f_{{\text{PSL}}}}\;\;\text{s.t.}\;\;(71) - (74),
\end{equation}
which is a convex quadratic programming (QP) that can be efficiently solved via off-the-shelf numerical tools.

\begin{figure}[!t]
	\centering
	\includegraphics[width = \columnwidth]{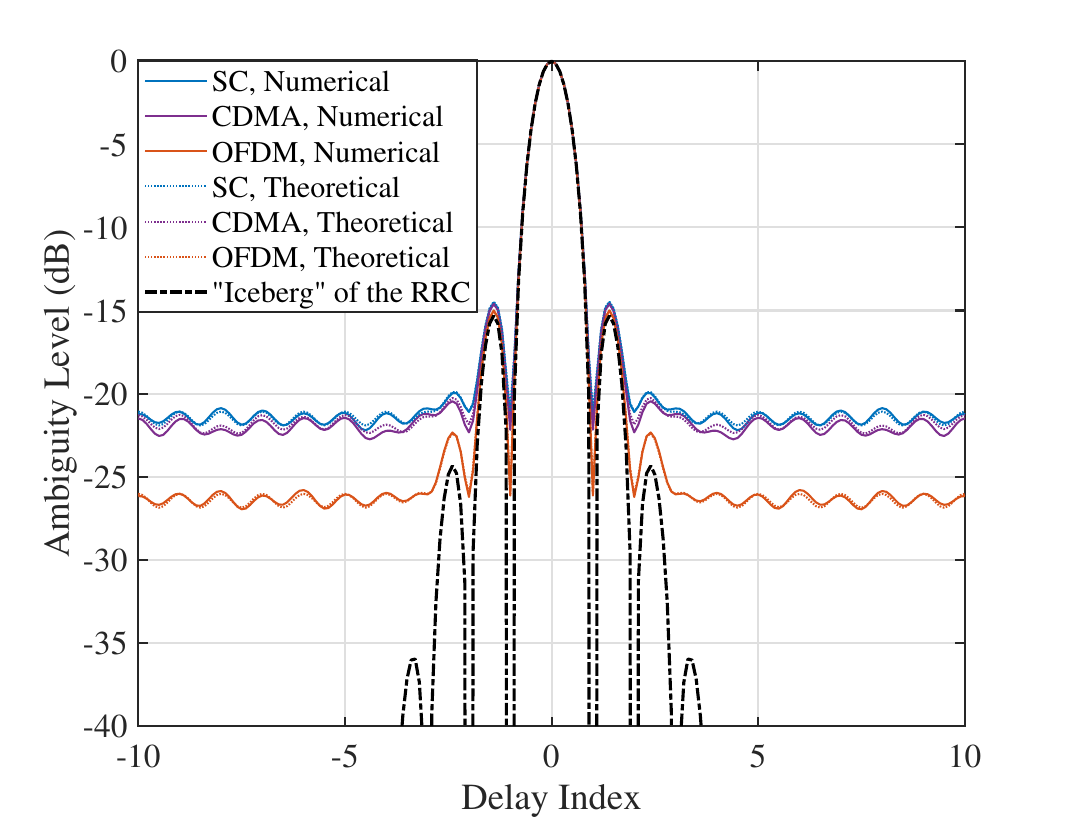}
	\caption{The average squared ACF of SC, CDMA, and OFDM signals, with 16-QAM constellation and $\alpha = 0.35$ RRC pulse shaping, $N = 128$, $L = 10$.}
    \label{SC_CDMA_OFDM}
\end{figure} 

\begin{figure}[!t]
	\centering
	\includegraphics[width = \columnwidth]{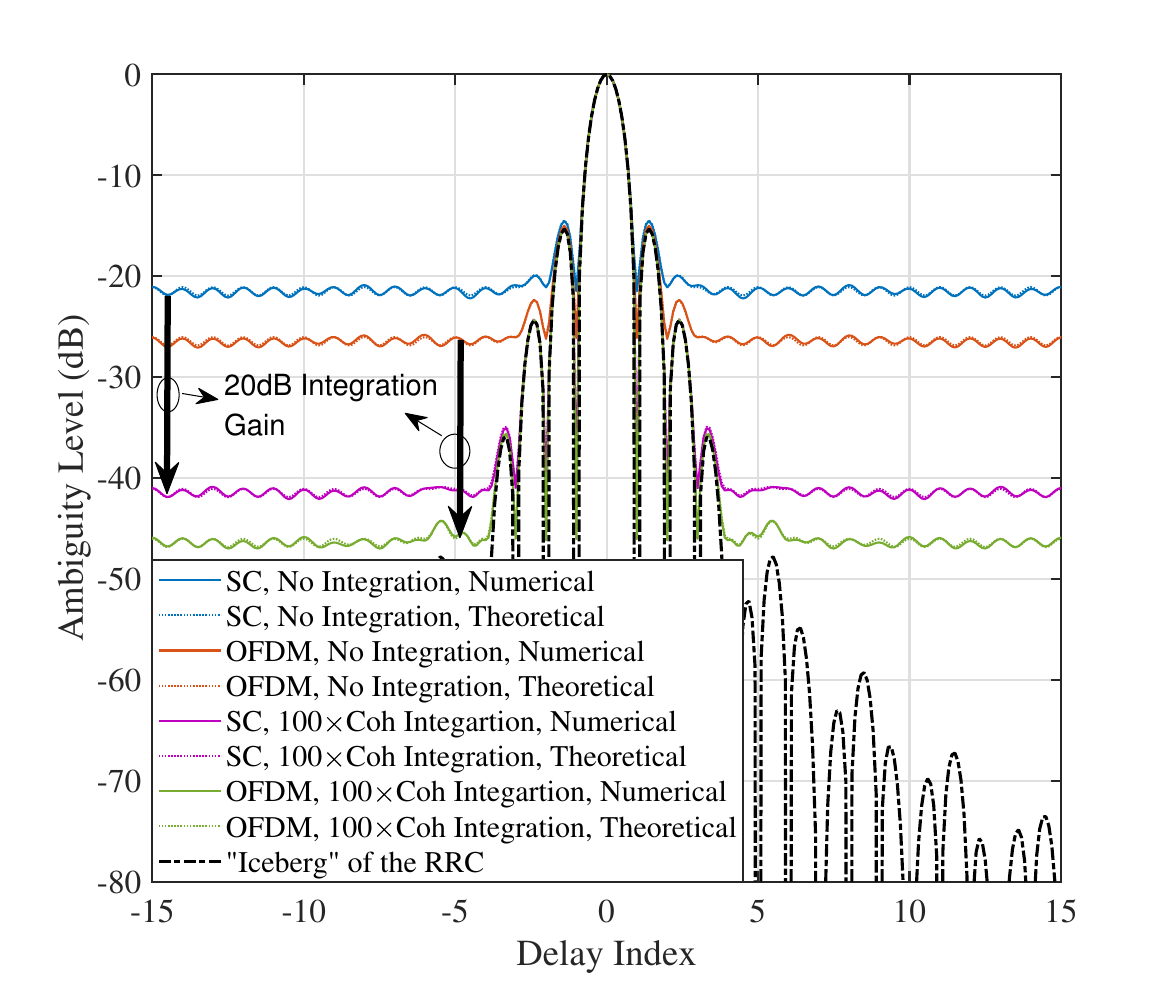}
	\caption{The average squared ACF of SC and OFDM signals under $M = 100$ coherent integration, with 16-QAM constellation and $\alpha = 0.35$ RRC pulse shaping, $N = 128$, $L = 10$.}
    \label{SC_OFDM_Coh}
\end{figure} 

\section{Numerical Results}\label{sec_5}
In this section, we present numerical results to validate the theoretical framework proposed in this paper. Unless otherwise specified, we consider an ISAC transmission of $N = 128$ random symbols under an $L = 10$ over-sampling ratio, utilizing Nyquist pulse shaping filters with a roll-off factor $\alpha = 0.35$. 

\subsection{Modulation Basis and Constellation Design}
We first illustrate the ACF performance across multiple modulation schemes. In Fig. \ref{SC_CDMA_OFDM}, we present the average squared ACF of SC, CDMA, and OFDM signals using a 16-QAM constellation and RRC pulse shaping with a roll-off factor of $\alpha = 0.35$. The CDMA modulation basis is generated using a size-$N$ Hadamard matrix. The results show that the theoretical performance aligns closely with the numerical results. Furthermore, as indicated by the theoretical findings, OFDM exhibits the lowest sidelobe level at every lag $k$ among the three modulation types, achieving a 5 dB reduction in sidelobe level compared to both SC and CDMA signals. We also depict the coherent integration performance of both OFDM and SC signaling in Fig. \ref{SC_OFDM_Coh}, using the same parameters as in Fig. \ref{SC_CDMA_OFDM} with a coherent integration count of $M = 100$. It is evident that a 20 dB reduction in sidelobe level is achievable for both SC and OFDM signals, with OFDM consistently delivering a 5 dB performance advantage over SC.

\begin{figure}[!t]
	\centering
	\includegraphics[width = \columnwidth]{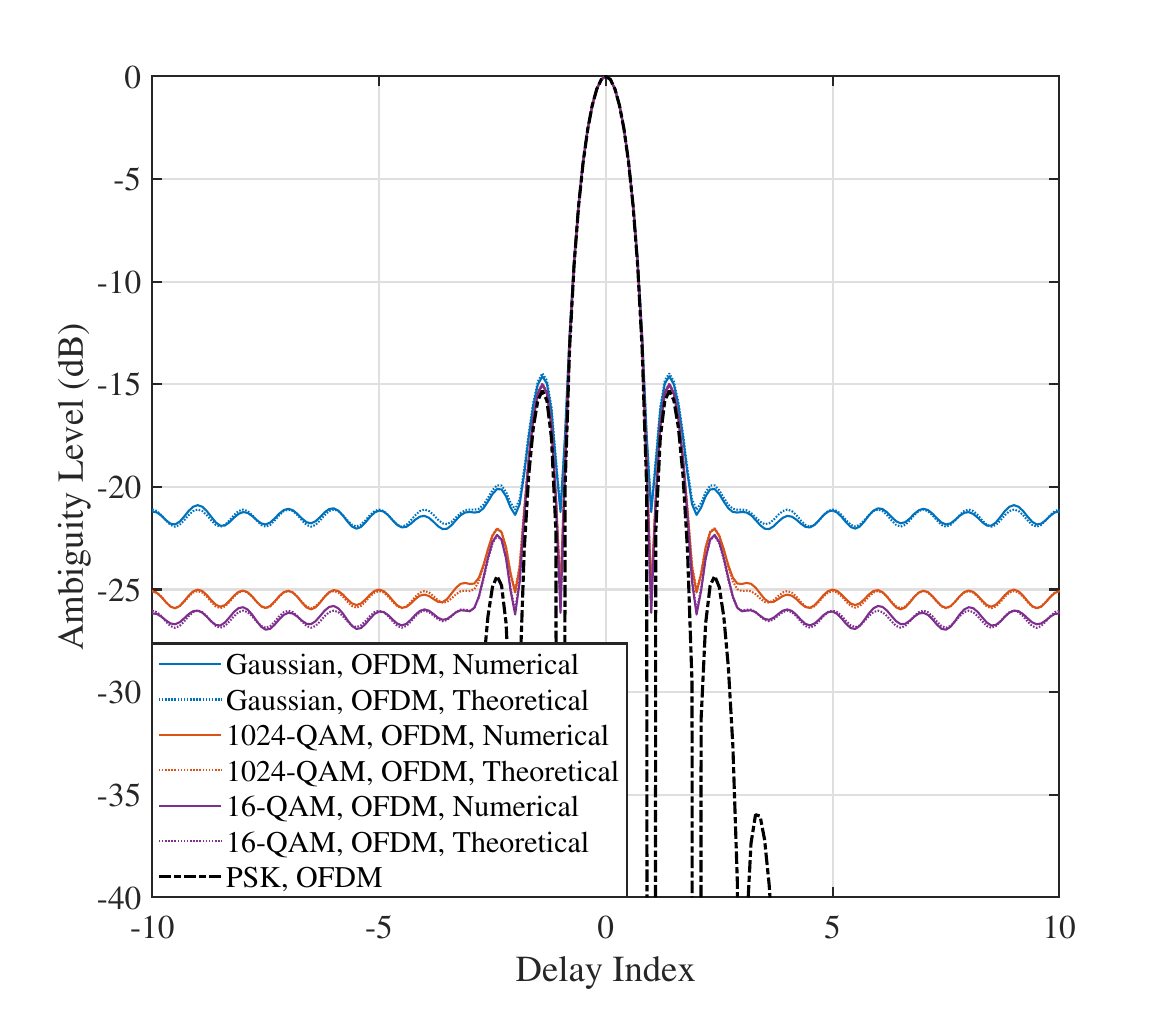}
	\caption{The average squared ACF of OFDM signals under PSK, 16-QAM, 1024-QAM, and Gaussian constellations under $\alpha = 0.35$ RRC pulse shaping, $N = 128$, $L = 10$.}
    \label{OFDM_constellations}
\end{figure} 

\begin{figure}[!t]
\centering
\subfloat[]{\includegraphics[width=0.5\columnwidth]{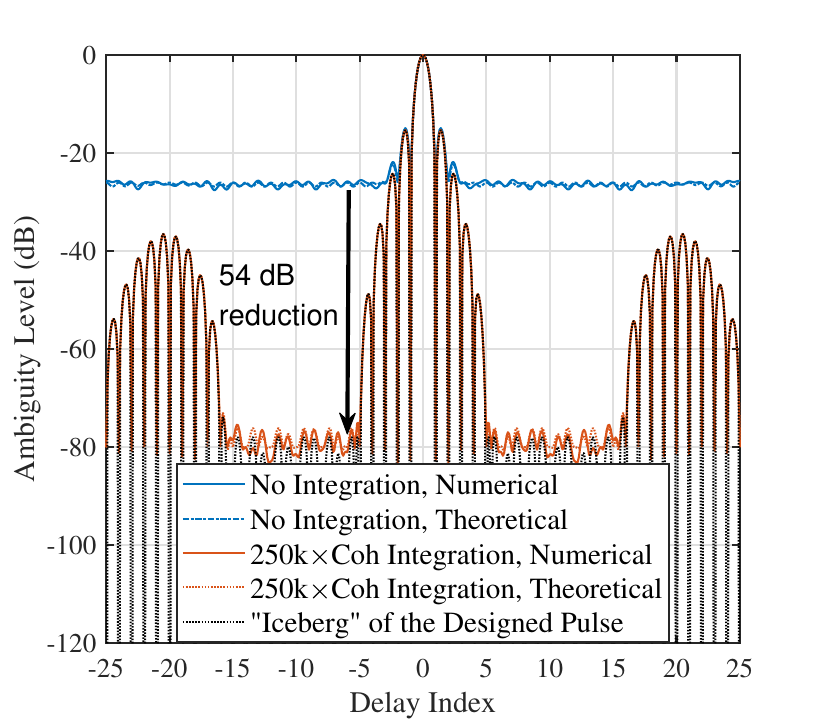}
\label{Iceberg_Coh}}
\subfloat[]{\includegraphics[width=0.5\columnwidth]{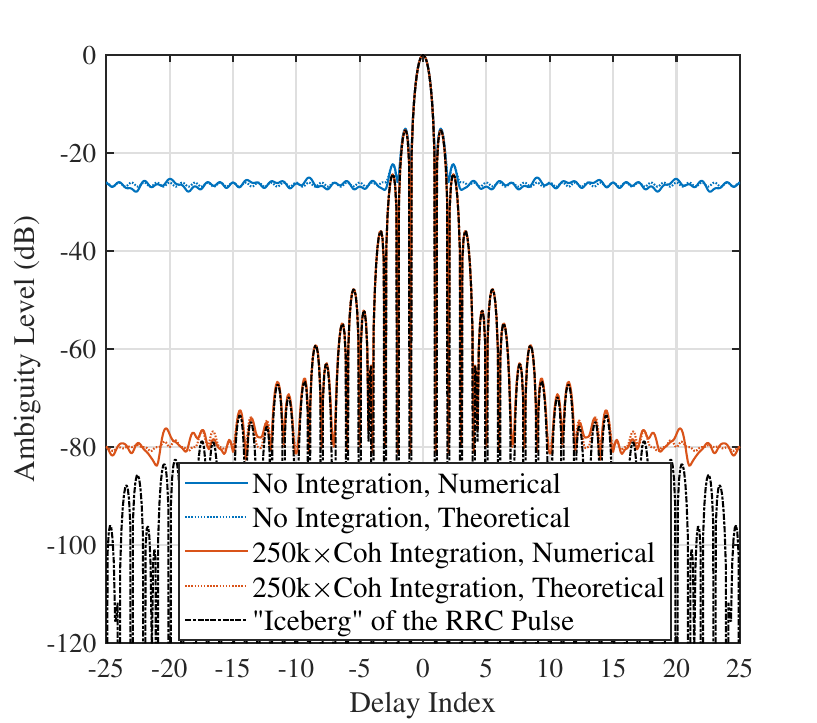}
\label{RRC_Coh}}
\vspace{0.1in}
\subfloat[]{\includegraphics[width=0.5\columnwidth]{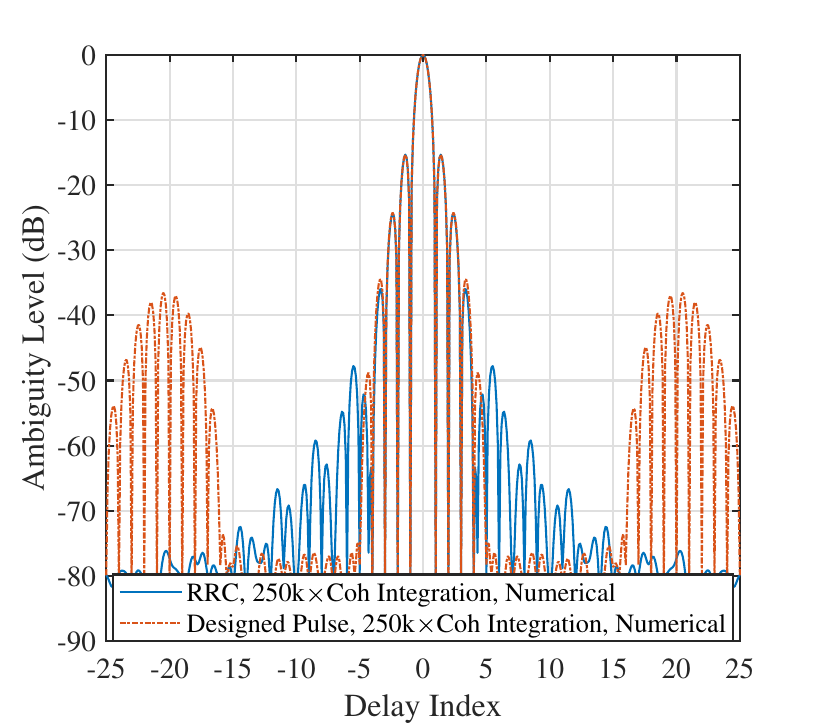}
\label{Iceberg_RRC_direct_comparison}}
\subfloat[]{\includegraphics[width=0.5\columnwidth]{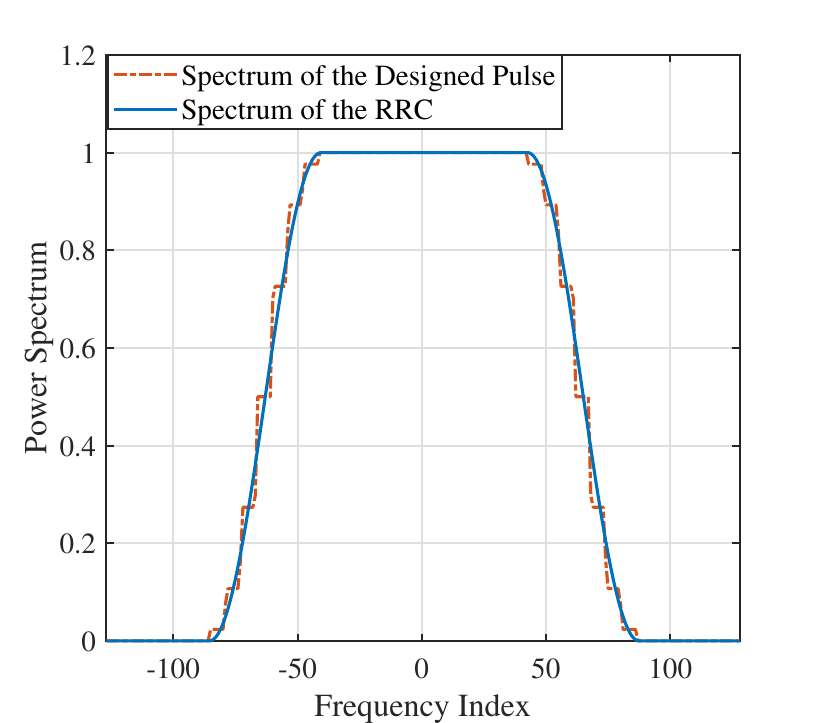}
\label{Iceberg_RRC_spectra_direct_comparison}}
\caption{The average squared ACFs with 250,000 coherent integrations under OFDM signaling with 16-QAM constellation, $N = 128$, $\alpha = 0.35$, $L = 10$. (a) Iceberg Shaping approach under PSL objective function, with the delay region of interest being $k \in \left[5,15\right]$. (b) RRC pulse shaping. (c) Direct comparison of coherently integrated ACFs under the proposed iceberg shaping and RRC pulse shaping using OFDM modulation and 16-QAM constellation. (d) The squared spectra of the designed pulse and the RRC.}
\label{Iceberg_RRC_Comparison}
\end{figure}

\begin{figure}[!t]
\centering
\subfloat[Range estimation performance.]{\includegraphics[width=0.92\columnwidth]{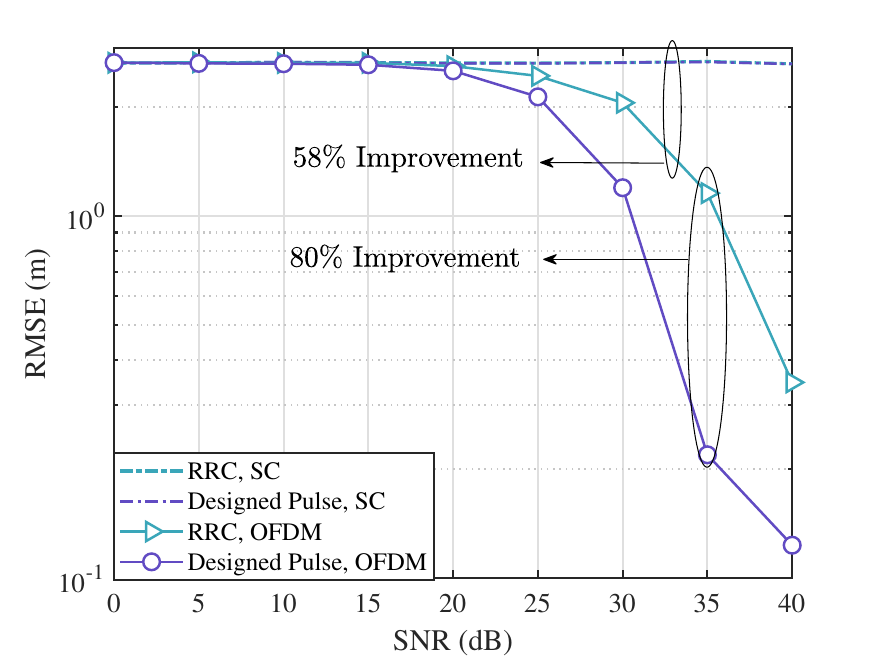}
\label{PSK_OFDM_vs_SC}}
\vspace{0.1in}
\subfloat[Range profiles.]{\includegraphics[width=1\columnwidth]{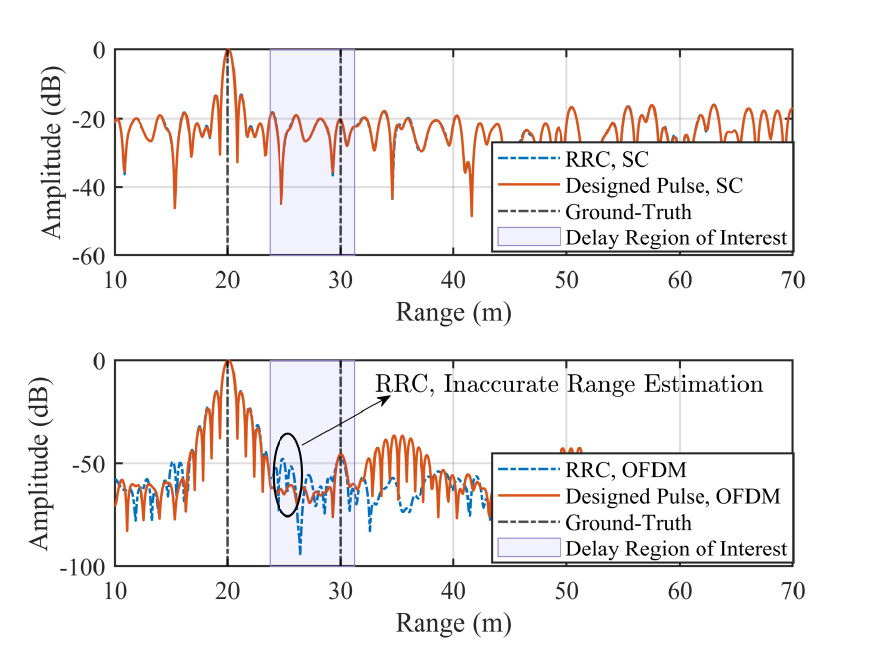}
\label{PSK_RangeProfile-eps-converted-to.pdf}}
\caption{The range estimation performance and profiles of two targets under SC and OFDM with 16-PSK constellation, where $N = 128$, $\alpha = 0.35$, $L = 10$, and range region of interest is $\left[23.74\text{m},31.24\text{m}\right]$.}
\label{PSK_QAM_RangeEstimation}
\end{figure}

\begin{figure}[!t]
\centering
\subfloat[Range estimation performance with/without coherent integration.]{\includegraphics[width=0.92\columnwidth]{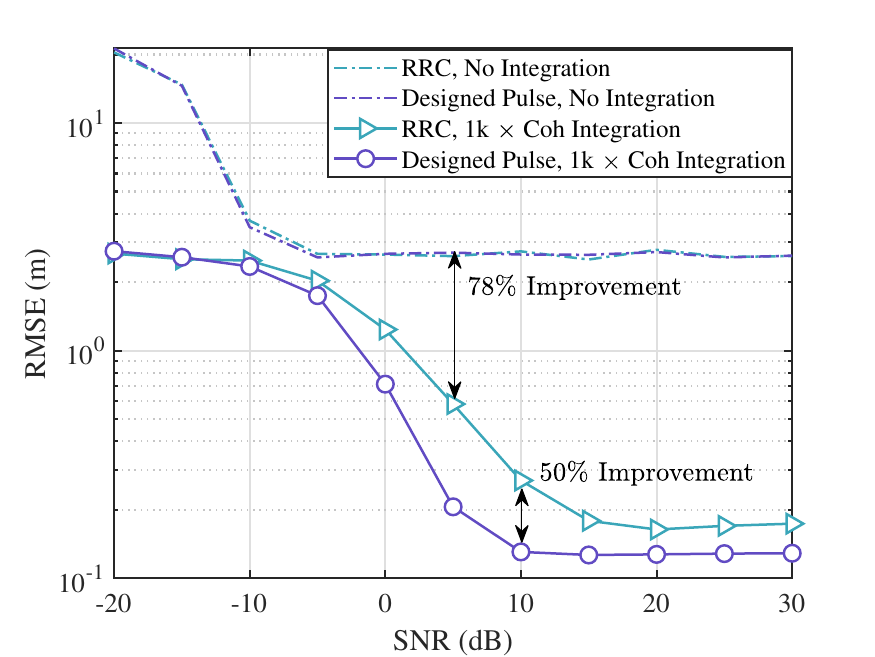}
\label{QAM_vs_NcohIntegration}}
\vspace{0.1in}
\subfloat[Range profiles with/without coherent integration.]{\includegraphics[width=1\columnwidth]{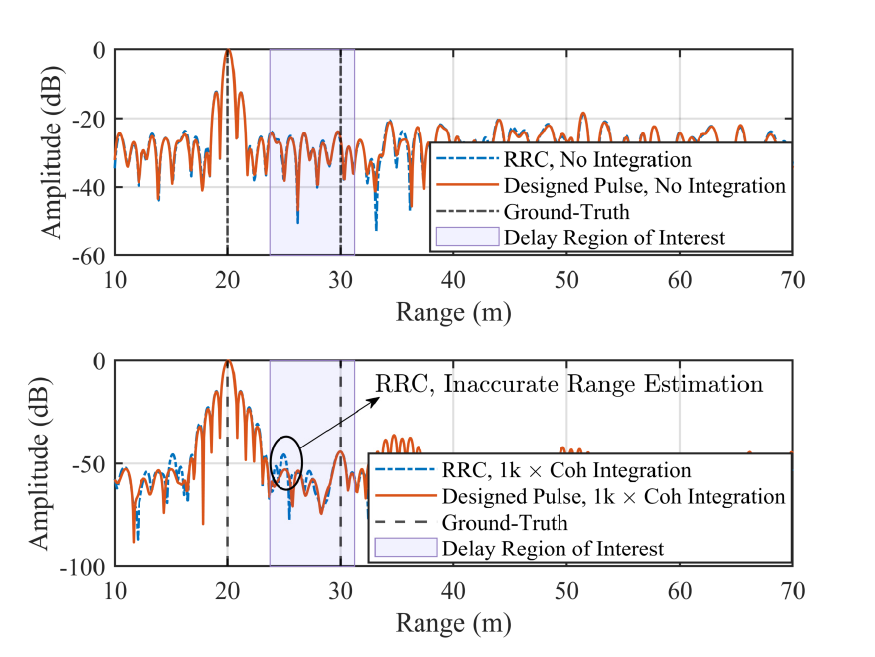}
\label{QAM_RangeProfile}}
\caption{The range estimation performance and profiles of two targets under OFDM with 16-QAM constellation, where $N = 128$, $\alpha = 0.35$, $L = 10$, $M = 1000$, and range region of interest is $\left[23.74\text{m},31.24\text{m}\right]$.}
\label{PSK_QAM_RangeEstimation}
\end{figure}

We then investigate the ACF of random OFDM ISAC signals with different constellation formats in Fig. \ref{OFDM_constellations}, under $\alpha = 0.35$ RRC pulse shaping. As expected, PSK constellations produce an average squared ACF identical to that of the RRC pulse itself, effectively representing the ``iceberg'' of the RRC. This indicates that PSK can successfully ``drain the sea'' when modulated by the OFDM basis. As the order of the QAM constellations increases, the sidelobe level rises due to the corresponding increase in the kurtosis. Finally, Gaussian constellations, with a kurtosis of 2, result in the highest ``sea level'', despite achieving the highest communication rate under Gaussian channels among all constellation types. This observation highlights the fundamental deterministic-random tradeoff in ISAC systems.

\subsection{Pulse Shaping Design}
Next, we examine the pulse shaping design of random ISAC signals by leveraging the iceberg theory proposed in this paper. As discussed in Sec. \ref{iceberg_shaping}, the ultimate ranging performance of random ISAC signals solely relies on the ACF of the pulse, provided that a sufficient number of MF outputs are coherently integrated, motivating the iceberg shaping approach. Here we study an example to show how iceberg shaping improves the performance of the ACF. In particular, we are interested in designing a pulse to reduce the sidelobe level within a given delay index region $k\in\left[5,15\right]$ by imposing a PSL objective function shown in \eqref{PSL}, with a roll-off factor $\alpha = 0.35$. The RRC pulse with the same $\alpha$ serves as the baseline.

The squared ACFs of the designed pulse and its RRC counterpart are illustrated as black dashed lines in Fig. \ref{Iceberg_Coh} and Fig. \ref{RRC_Coh}, respectively. Notably, the designed pulse achieves significant suppression in the sidelobe region for $k\in\left[5,15\right]$, resulting in a uniform sidelobe level of $-80$ dB. To attain this low sidelobe level, the number of coherent integrations must reach to $M = 250,000$, leading to a 54 dB reduction in the ``sea level'' against to the no-integration case. For further clarity, we directly compare the sidelobe levels of $M = 250,000$ coherently integrated ACFs for both strategies in Fig. \ref{Iceberg_RRC_direct_comparison}. This comparison highlights an up to 30 dB reduction in sidelobe level within $k\in\left[5,15\right]$ for the proposed iceberg shaping method, relative to the RRC benchmark, albeit at the cost of increased sidelobe power in other regions.

Finally, we present the resulting squared spectra of the designed pulse and the RRC in Fig. \ref{Iceberg_RRC_spectra_direct_comparison}, both of which strictly adhere to the folded spectrum criterion. Interestingly, the designed pulse displays a stepped shape in its frequency spectrum, despite having the same roll-off factor as the RRC. This indicates that the degrees of freedom for designing the Nyquist pulse primarily reside in the roll-off portion.


\subsection{Ranging Performance Analysis}
In this subsection, we present concrete ranging results for both SC and OFDM modulations under both RRC and designed pulses, highlighting the effectiveness of the proposed iceberg shaping and coherent integration techniques. Here we consider the transmission of $N = 128$ symbols over a bandwidth of $200$ MHz, translating to a subcarrier spacing of $1.5625$ MHz for OFDM signaling. We examine the detection performance of two targets located at $20$m (strong target) and $30$m (weak target), respectively. Accordingly, the iceberg shaping technique conceives a Nyquist pulse by minimizing the ISL within the delay region $\left[23.74\text{m},31.24\text{m}\right]$. The strong target has an amplitude that is $43\sim46$ dB higher than that of the weak targets. All results are attained from averaging over $1000$ Monte Carlo simulations.

As illustrated in Fig. \ref{PSK_OFDM_vs_SC}, we compare the ranging performance of the RRC and the proposed pulse shaping methods under SC and OFDM with the 16-PSK constellation. It is found that OFDM achieves significantly better performance compared to SC modulation, with more than $58\%$ improvement at high SNRs. This is because SC modulation exhibits a higher ``sea level'', which may obscure the reflection amplitude of weak targets in the range profile, as shown in Fig. \ref{PSK_RangeProfile.eps}. Furthermore, the proposed iceberg shaping technique shows an $80\%$ improvement at SNR = $35$ dB compared to that of the RRC. As depicted in Fig. \ref{PSK_RangeProfile.eps}, the RRC pulse shaping method may fail to identify the accurate range peak in the delay region of interest, as it is masked by the strong target's sidelobes. In contrast, the proposed pulse shaping scheme enables accurate range estimation thanks to its ability to minimize the ISL in the delay region of interest.

To validate the superior performance of the proposed iceberg shaping technique under coherent integration, we present the range estimation results for OFDM with a 16-QAM constellation in Fig. \ref{QAM_vs_NcohIntegration}. It is observed that OFDM signals without integration may fail to estimate the range of weak targets, regardless of the chosen pulse shaping method. This is because the reflection echo amplitude of the weak target is entirely masked by the sidelobes of the strong target, as shown in Fig. \ref{QAM_RangeProfile}. After $1000$ times coherent integration, both the RRC pulse and iceberg shaped pulse are able to identify the peak of the weak target. Notably, the proposed pulse shaping method achieves a $50\%$ performance improvement compared to that of the RRC. This demonstrates that the proposed iceberg shaping method, combined with coherent integration, supports flexible and reliable ranging in multi-target scenarios.

\section{Conclusion}\label{sec_6}
This study has provided a comprehensive analysis of the ranging performance of communication-centric ISAC signals, with a particular focus on modulation and pulse shaping design to improve target detection performance. By exploring how random data payload signals can be leveraged for both communication and sensing, we sought to reshape the statistical properties of auto-correlation functions (ACFs) to enhance sensing capabilities. We derived a closed-form expression for the expectation of the squared ACF of random ISAC signals, considering arbitrary modulation bases and constellation mappings within the Nyquist pulse shaping framework. Our analysis introduced a metaphorical ``iceberg hidden in the sea'' structure, where the ``iceberg'' represents the squared mean of the ACF, corresponding to the squared ACF of the adopted pulse shaping filter, while the ``sea level'' , arises from the randomness of data payloads, characterizes the variance of the ACF. Our results demonstrated that, for QAM/PSK constellations with Nyquist pulse shaping, OFDM achieves the lowest ranging sidelobe levels across all lags. Inspired by these findings, we proposed a novel Nyquist pulse shaping design to further enhance the sensing performance of random ISAC signals. Numerical evaluations validated our theoretical findings, confirming that the proposed pulse shaping significantly reduces ranging sidelobes when compared to conventional root-raised cosine (RRC) pulse shaping. These insights provide a promising pathway for optimizing ISAC systems in future 6G networks, improving the sensing functionality without sacrificing the communication performance.

\appendices
\section{Proof of Theorem 1}\label{Theorem_1_proof}
Let us first express the squared ACF as
    \begin{align}
        |{R}_k|^2 \nonumber&= \sum\limits_{n = 1}^{N}\tilde{g}_{n,k}|\mathbf{v}_n^H\mathbf{s}|^2 e^{\frac{j2\pi k(n-1)}{LN}}\sum\limits_{m = 1}^{N}\tilde{g}_{m,k}^\ast|\mathbf{v}_m^H\mathbf{s}|^2 e^{\frac{-j2\pi k(m-1)}{LN}}\\
        &=\sum\limits_{n = 1}^{N}\sum\limits_{m = 1}^{N}\tilde{g}_{n,k}\tilde{g}_{m,k}^\ast|\mathbf{v}_n^H\mathbf{s}|^2|\mathbf{v}_m^H\mathbf{s}|^2e^{\frac{j2\pi k(n-m)}{LN}}.
    \end{align}
Moreover, note that
    \begin{align}
        |\mathbf{v}_n^H\mathbf{s}|^2 \nonumber&= \mathbf{v}_n^H\mathbf{s}\mathbf{s}^H\mathbf{v}_n= (\mathbf{v}_n^T \otimes \mathbf{v}_n^H)\operatorname{vec}(\mathbf{s}\mathbf{s}^H)\\
        &=(\mathbf{v}_n^T \otimes \mathbf{v}_n^H)\tilde{\mathbf{s}} = \tilde{\mathbf{s}}^H(\mathbf{v}_n^\ast \otimes \mathbf{v}_n),
    \end{align}
yielding
    \begin{align}\label{average_P-ACF_step_1}
        &\nonumber \mathbb{E}(|R_k|^2) = \\
        &\sum\limits_{n = 1}^{N}\sum\limits_{m = 1}^{N}\tilde{g}_{n,k}\tilde{g}_{m,k}^\ast(\mathbf{v}_n^T \otimes \mathbf{v}_n^H)\mathbf{S}(\mathbf{v}_m^\ast \otimes \mathbf{v}_m)e^{\frac{j2\pi k(n-m)}{LN}},
    \end{align}
where $\mathbf{S} = \mathbb{E}(\tilde{\mathbf{s}}\tilde{\mathbf{s}}^H)$. To proceed, we exploit the following lemma to simplify the computation.
\begin{lemma}
    Let $\tilde{\mathbf{s}} = \operatorname{vec}(\mathbf{s}\mathbf{s}^H)$. For constellations that meet the Assumptions 1 and 2, we have
    \begin{align}
        \mathbf{S} \nonumber&= \mathbb{E}(\tilde{\mathbf{s}}\tilde{\mathbf{s}}^H) \\
        &= \left[ {\begin{array}{*{20}{c}}
  {{\mu _4}}&{{\mathbf{0}}_N^T}&1&{{\mathbf{0}}_N^T}&1& \cdots &1 \\ 
  {{{\mathbf{0}}_N}}&{{{\mathbf{I}}_N}}&{{{\mathbf{0}}_N}}&{{{\mathbf{0}}_N}}&{{{\mathbf{0}}_N}}& \cdots &{{{\mathbf{0}}_N}} \\ 
  1&0&{{\mu _4}}& \ldots &1& \ldots &1 \\ 
  {{{\mathbf{0}}_N}}&{{{\mathbf{0}}_N}}&{{{\mathbf{0}}_N}}&{{{\mathbf{I}}_N}}& \vdots & \cdots & \vdots  \\ 
  1&0&1&0&{{\mu _4}}& \cdots &1 \\ 
   \vdots & \vdots & \vdots & \vdots & \vdots & \ddots &{{{\mathbf{0}}_N}} \\ 
  1&0&1&0&1& \cdots &{{\mu _4}} 
\end{array}} \right] \in \mathbb{R}^{N^2\times N^2},
    \end{align}
    where $\mu_4$ is the kurtosis of the constellation, and $\mathbf{0}_{N}$ represents the all-zero vector with length N.
\end{lemma}
\begin{proof}
See \cite{liu2024ofdm}.
\end{proof}

Let us further decompose $\mathbf{S}$ as
    \begin{equation}\label{S_decompose}
        \mathbf{S} = \mathbf{I}_{N^2} + \mathbf{S}_1 + \mathbf{S}_2,
    \end{equation}
where 
    \begin{align}
        {{\mathbf{S}}_1} &= \operatorname{Diag} \left( {{{\left[ {{\mu _4}-2,{\mathbf{0}}_N^T,{\mu _4}-2,{\mathbf{0}}_N^T, \ldots {\mu _4}-2} \right]}^T}} \right),\\
        {{\mathbf{S}}_2} &= \left[ {{\mathbf{c}},{{\mathbf{0}}_{{N^2} \times N}},{\mathbf{c}}, \ldots ,{\mathbf{c}},{{\mathbf{0}}_{{N^2} \times N}},{\mathbf{c}}} \right],
    \end{align}
with ${{\mathbf{0}}_{{N^2} \times N}}$ being the all-zero matrix of size ${{N^2} \times N}$, and 
    \begin{equation}
        {\mathbf{c}} = {\left[ {1,{\mathbf{0}}_N^T,1, \ldots ,1,{\mathbf{0}}_N^T,1} \right]^T}.
    \end{equation}
By leveraging $\mathbf{v}_n^H\mathbf{v}_m = \delta_{n,m}$, we have
    \begin{align}
        &(\mathbf{v}_n^T \otimes \mathbf{v}_n^H)\mathbf{I}_{N^2}(\mathbf{v}_m^\ast \otimes \mathbf{v}_m) = \mathbf{v}_n^T\mathbf{v}_m^\ast\mathbf{v}_n^H\mathbf{v}_m = \delta_{n,m}, \label{1st_term}\\\nonumber
        &(\mathbf{v}_n^T \otimes \mathbf{v}_n^H)\mathbf{S}_1(\mathbf{v}_m^\ast \otimes \mathbf{v}_m)=  (\mu_4 -2)\sum\limits_{p = 1}^{N}|v_{p,n}|^2|v_{p,m}|^2 \\
        &=(\mu_4 -2)\left\| {{{\mathbf{v}}_n} \odot {{\mathbf{v}}_m}} \right\|^2, \label{2nd_term}\\\nonumber
        &(\mathbf{v}_n^T \otimes \mathbf{v}_n^H)\mathbf{S}_2(\mathbf{v}_m^\ast \otimes \mathbf{v}_m)=\sum\limits_{p = 1}^{N}|v_{p,n}|^2\sum\limits_{q = 1}^{N}|v_{q,m}|^2 \\
        & = \left\| {{\mathbf{v}}_n} \right\|^2\left\| {{\mathbf{v}}_m} \right\|^2 = 1,\label{3rd_term}
    \end{align}
Plugging \eqref{S_decompose}, \eqref{1st_term}-\eqref{3rd_term} into \eqref{average_P-ACF_step_1} yields
    \begin{align}\label{ERk2}
        &\nonumber\mathbb{E}(|R_k|^2) =  {\left| {\sum\limits_{n = 1}^N {{{\tilde g}_{n,k}}{e^{\frac{{j2\pi \left( {n - 1} \right)k}}{{LN}}}}} } \right|^2} + \sum\limits_{n = 1}^N {{{\left| {{{\tilde g}_{n,k}}} \right|}^2}}  \\ & + \left( {{\mu _4} - 2} \right)\sum\limits_{n = 1}^N {\sum\limits_{m = 1}^N {{{\tilde g}_{n,k}}{\tilde g}_{m,k}^\ast} } \left\| {{{\mathbf{v}}_n} \odot {{\mathbf{v}}_m}} \right\|^2{e^{\frac{{j2\pi \left( {n - m} \right)k}}{{LN}}}}\\
        &  = N{\left| {\tilde{\mathbf{f}}_{k + 1}^H\tilde {\mathbf{g}}_k} \right|^2} + \left\| {\tilde {\mathbf{g}}_k} \right\|^2 + ({\mu _4} - 2)N\left\| {\tilde {\mathbf{V}}\left( {\tilde {\mathbf{g}}_k\odot {\tilde{\mathbf{f}}}_{k + 1}^\ast} \right)} \right\|^2.
    \end{align}
Moreover, note that the mean of $R_k$ can be expressed as
    \begin{align}\label{average_Rk}
        \mathbb{E}\left( {R_k}\right)&\nonumber = \sum\limits_{n = 1}^{N}\tilde{g}_{n,k}\mathbb{E}\left(|\mathbf{v}_n^H\mathbf{s}|^2\right) e^{\frac{j2\pi k(n-1)}{LN}}\\
        & = \sum\limits_{n = 1}^{N}\tilde{g}_{n,k} e^{\frac{j2\pi k(n-1)}{LN}} = \sqrt{N}\tilde{\mathbf{f}}_{k + 1}^H\tilde {\mathbf{g}},
    \end{align}
leading to the variance of $R_k$ in the form of
    \begin{align}\label{var_Rk}
        \operatorname{var}(R_k) &\nonumber = \mathbb{E}(|R_k|^2) -  \left|\mathbb{E}\left( {R_k}\right)\right|^2\\
        & = \left\| {\tilde {\mathbf{g}}_k} \right\|^2 + ({\mu _4} - 2)N\left\| {\tilde {\mathbf{V}}\left( {\tilde {\mathbf{g}}_k\odot {\tilde{\mathbf{f}}}_{k + 1}^\ast} \right)} \right\|^2.
    \end{align}
This concludes the proof.

\section{Proof of Theorem 2}\label{Theorem_2_Proof}
When $\mu_4 < 2$, minimizing the sea level in the sidelobe region is equivalent to maximize $\left\| {\tilde {\mathbf{V}}\left( {\tilde {\mathbf{g}}_k\odot {\mathbf{f}}_{k + 1}^\ast} \right)} \right\|^2$ over the nonconvex set of unistochastic matrices, which may be relaxed to the following problem:
\begin{equation}\label{cvx_relax}
    \mathop {\max }\limits_{\tilde{\mathbf{V}} \in \mathbb{V}} \;\;\left\| {\tilde {\mathbf{V}}\left( {\tilde {\mathbf{g}}_k\odot \tilde{\mathbf{f}}_{k + 1}^\ast} \right)} \right\|^2,
\end{equation}
where $\mathbb{V}$ represents the set of bistochastic matrices, which is the convex hull of the set of unistochastic matrices. We will show in the sequel that such a convex relaxation is tight, i.e., the optimal solution to \eqref{cvx_relax} is also the optimal solution to maximizing the objective function over the set of unistochastic matrices.

Let us decompose ${\tilde {\mathbf{g}}_k\odot \tilde{\mathbf{f}}_{k + 1}^\ast}$ into
\begin{equation}
    {\tilde {\mathbf{g}}_k\odot \tilde{\mathbf{f}}_{k + 1}^\ast} = \mathbf{b}_R + j\mathbf{b}_I,
\end{equation}
where $\mathbf{b}_R$ and $\mathbf{b}_I$ are the real and imaginary parts of ${\tilde {\mathbf{g}}_k\odot \tilde{\mathbf{f}}_{k + 1}^\ast}$. We then have
\begin{equation}\label{real_decomposition}
   \left\| {\tilde {\mathbf{V}}\left( {\tilde {\mathbf{g}}_k\odot \tilde{\mathbf{f}}_{k + 1}^\ast} \right)} \right\|^2 = \left\| {\tilde {\mathbf{V}}\mathbf{b}_R} \right\|^2 + \left\| {\tilde {\mathbf{V}}\mathbf{b}_I} \right\|^2.
\end{equation}
We prove Theorem 2 by the concept of majorization. Given a pair of real vectors $\mathbf{a},\mathbf{b}\in\mathbb{R}^{N \times 1}$, let $a_n^{\downarrow}, b_n^{\downarrow}$ be their $n$th largest entries, respectively. If
\begin{align}
    &\sum\limits_{n = 1}^N a_n = \sum\limits_{n = 1}^N b_n, \;\sum\limits_{n = 1}^k a_n^{\downarrow} \ge  \sum\limits_{n = 1}^k b_n^{\downarrow}, \quad k = 1,2,\ldots,N,
\end{align}
we say that $\mathbf{a}$ majorizes $\mathbf{b}$, denoted as $\mathbf{a}\succ\mathbf{b}$. This may be equivalently defined as $\mathbf{b} = \tilde{\mathbf{V}}\mathbf{a}$, where $\tilde{\mathbf{V}}$ is any bistochastic matrix. Consequently, it follows that ${\tilde{\mathbf{V}}\mathbf{b}_R} \prec \mathbf{b}_R$ and ${\tilde{\mathbf{V}}\mathbf{b}_I} \prec \mathbf{b}_I$. Recognizing that the $\ell_2$ norm is a Schur-convex function, which is monotonic with respect to majorization, we have
\begin{equation}
    \left\|{\tilde {\mathbf{V}}\mathbf{b}_R}\right\|^2 \le \left\|\mathbf{b}_{R}\right\|^2, \quad \left\|{\tilde {\mathbf{V}}\mathbf{b}_I}\right\|^2 \le \left\|\mathbf{b}_{I}\right\|^2,
\end{equation}
and hence
\begin{equation}
    \left\| {\tilde {\mathbf{V}}\left( {\tilde {\mathbf{g}}_k\odot \tilde{\mathbf{f}}_{k + 1}^\ast} \right)} \right\|^2 \le  \left\| {\tilde {\mathbf{g}}_k\odot \tilde{\mathbf{f}}_{k + 1}^\ast} \right\|^2,
\end{equation}
where all equalities hold \textit{if and only if} $\tilde{\mathbf{V}} = \mathbf{V}\odot\mathbf{V}^\ast$ is a permutation matrix, which is obviously unistochastic. This suggests that the only possible form of the optimal $\mathbf{V}$ is a complex permutation matrix, i.e.,
\begin{equation}
    \mathbf{V}_{\text{sub}} = \bm{\Pi}\operatorname{Diag}({\bm{{\theta}}}),
\end{equation}
where $\bm{\Pi}$ is any permutation matrix, and ${\bm{{\theta}}}\in\mathbb{C}^{N\times 1}$ is an arbitrary unit-modulus vector representing the phases. Accordingly, the optimal signaling basis has the form of 
\begin{equation}\label{ofdm_opt}
    \mathbf{U}_{\text{sub}} = \mathbf{F}_N^H\operatorname{Diag}({\bm{{\theta}}}^\ast)\bm{\Pi}^T.
\end{equation}
 If $\bm{\Pi} = \mathbf{I}_N, {\bm{{\theta}}} = \mathbf{1}_N$, then $\mathbf{U}_{\text{sub}}$ represents the standard OFDM modulation. Otherwise, \eqref{ofdm_opt} simply results in an OFDM modulation with different initial phases over permuted subcarriers. This also confirms that OFDM is the only modulation basis achieving the lowest ranging sidelobe, due to the uniqueness of $\mathbf{V}_{\text{sub}}$.

 \section{Proof of Theorem 3} \label{Theorem_3_proof}
 When $\mu_4 > 2$, minimizing the sea level in the sidelobe region is to minimize $\left\| {\tilde {\mathbf{V}}\left( {\tilde {\mathbf{g}}_k\odot \tilde{\mathbf{f}}_{k + 1}^\ast} \right)} \right\|^2$ over the set of unistochastic matrices, which can be similarly relaxed into
\begin{equation}\label{cvx_relax_min}
    \mathop {\min }\limits_{\tilde{\mathbf{V}} \in \mathbb{V}} \;\;\left\| {\tilde {\mathbf{V}}\left( {\tilde {\mathbf{g}}_k\odot \tilde{\mathbf{f}}_{k + 1}^\ast} \right)} \right\|^2.
\end{equation}
Solving problem \eqref{cvx_relax_min} is equivalent to minimizing \eqref{real_decomposition}. Let $\mathbf{D} = \frac{1}{N}\mathbf{1}_{N}\mathbf{1}_{N}^T$, which is a uniform bistochastic matrix. For any bistochastic matrix ${\tilde{\mathbf{V}}}$, it holds $\mathbf{D}{\tilde{\mathbf{V}}} = \mathbf{D}$. This indicates that
\begin{equation}
    \mathbf{D}{\tilde{\mathbf{V}}}\mathbf{b}_{R} = \mathbf{D}\mathbf{b}_{R}, \quad \mathbf{D}{\tilde{\mathbf{V}}}\mathbf{b}_{I} = \mathbf{D}\mathbf{b}_{I}
\end{equation}
i.e., $\mathbf{D}\mathbf{b}_{R}\prec {\tilde{\mathbf{V}}}\mathbf{b}_{R}$, $\mathbf{D}\mathbf{b}_{I}\prec {\tilde{\mathbf{V}}}\mathbf{b}_{I}$. Again, due to the Schur-convexity of the $\ell_2$ norm, we have
\begin{equation}
    \left\|{\tilde{\mathbf{V}}}\mathbf{b}_{R}\right\|^2 \ge \left\|\mathbf{D}\mathbf{b}_{R}\right\|^2, \quad \left\|{\tilde{\mathbf{V}}}\mathbf{b}_{I}\right\|^2 \ge \left\|\mathbf{D}\mathbf{b}_{I}\right\|^2,
\end{equation}
and hence
\begin{equation}
    \left\| {\tilde {\mathbf{V}}\left( {\tilde {\mathbf{g}}_k\odot \tilde{\mathbf{f}}_{k + 1}^\ast} \right)} \right\|^2 \ge  \left\| {{\mathbf{D}}\left( {\tilde {\mathbf{g}}_k\odot \tilde{\mathbf{f}}_{k + 1}^\ast} \right)} \right\|^2,
\end{equation}
which suggests that the optimal $\mathbf{V}$ should be both unitary and have constant modulus. A proper choice would be the IDFT matrix, namely, 
\begin{equation}
    \mathbf{V}_{\text{super}} = \mathbf{F}_N^H,
\end{equation}
resulting in the optimal signaling basis for super-Gaussian constellations:
\begin{equation}
    \mathbf{U}_{\text{super}} = \mathbf{I}_N,
\end{equation}
which leads to an SC modulation.
	\bibliographystyle{IEEEtran}
	\bibliography{IEEEabrv,references_SPM,references,database}

\end{document}